%% file: dime.tex
\documentclass{svjour3}

\input{header.tex}

\input{alg.tex}

\begin{document}
\title{Schemas for Unordered XML on a DIME\thanks{A preliminary version of this article has appeared in the Proceedings of the 16th International Workshop on the Web and Databases (WebDB), 2013~\cite{BoCiSt13}.}}
\author{Iovka~Boneva \and Radu~Ciucanu \and S\l awek~Staworko}
\authorrunning{Iovka~Boneva, Radu~Ciucanu, S\l awek~Staworko}
\institute{University of Lille \& INRIA.
\email{\{iovka.boneva,radu.ciucanu,slawomir.staworko\}@inria.fr}.
\\
Parc scientifique de la Haute Borne - 40, avenue Halley - B\^at B - Park Plaza, 59650 Villeneuve d'Ascq - France.}

\maketitle
\thispagestyle{empty}
\begin{abstract}
We investigate schema languages for unordered XML having no relative order among siblings.
First, we propose \emph{unordered regular expressions} (UREs), essentially regular expressions with \emph{unordered concatenation} instead of standard concatenation, that define languages of unordered words to model the allowed content of a node (i.e., collections of the labels of children). 
However, unrestricted UREs are computationally too expensive as we show the intractability of two fundamental decision problems for UREs: membership of an unordered word to the language of a URE and containment of two UREs. 
Consequently, we propose a practical and tractable restriction of UREs, \emph{disjunctive interval multiplicity expressions} (DIMEs).

Next, we employ DIMEs to define languages of unordered trees and propose two schema languages: \emph{disjunctive interval multiplicity schema} (DIMS), and its restriction, \emph{disjunction-free interval multiplicity schema} (IMS). 
We study the complexity of the following static analysis problems: schema satisfiability, membership of a tree to the language of a schema, schema containment, as well as twig query satisfiability, implication, and containment in the presence of schema. 
Finally, we study the expressive power of the proposed schema languages and compare them with yardstick languages of unordered trees (FO, MSO, and Presburger constraints) and DTDs under commutative closure. 
Our results show that the proposed schema languages are capable of expressing many practical languages of unordered trees and enjoy desirable computational properties.
\end{abstract}


\section{Introduction}
When XML is used for \emph{document-centric} applications, the
relative order among the elements is typically important e.g., the relative order of paragraphs and chapters in a book. On the other hand, in case of \emph{data-centric} XML applications, the order among the elements may be unimportant~\cite{AbBoVi12}. 
In this paper we focus on the latter use case. 
As an example, take a trivialized fragment of an XML document containing the DBLP repository in Figure~\ref{fig:dblp}. 
While the order of the elements \rnew{title}, \rnew{author}, and \rnew{year} may differ from one publication to
another, it has no impact on the semantics of the data stored in this semi-structured database.
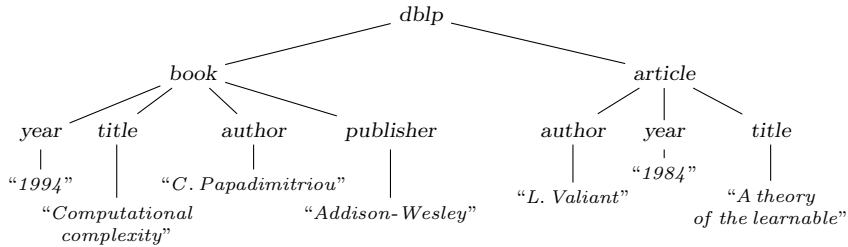
\begin{figure}[h]
  \centering
  \begin{tikzpicture}[yscale=0.75,xscale=2]\small
		\node at (0,0) (n0) {\sl dblp};
		\node at (-1.5,-1) (n10) {\sl book}edge[-] (n0);
		\node at (-2.5, -2.1) (m111) {\sl year}edge[-] (n10);
		\node at (-2, -2) (m112) {\sl title}edge[-] (n10);
		\node at (-1.1, -2) (m113) {\sl author}edge[-] (n10);
		\node at (-0.2, -2.05) (m114) {\sl publisher}edge[-] (n10);
		\node at (-2.5,-3) {\scriptsize ``$\mathit{1994}$''} edge[-](m111);
		\node at (-2,-3.5) {\scriptsize ``$\mathit{Computational}$} edge[-] (m112);
		\node at (-2,-3.9) {\scriptsize $\mathit{complexity}$''};
		\node at (-1.1,-3) {\scriptsize ``$\mathit{C.\,Papadimitriou}$''} edge[-] (m113);
		\node at(-0.2, -3.5) {\scriptsize ``$\mathit{Addison}$-$\mathit{Wesley}$''} edge[-] (m114);
		\node at (1.6,-1) (n11) {\sl article}edge[-] (n0);
		\node at (1, -2) (m111) {\sl author}edge[-] (n11);
		\node at (1.6, -2.1) (m112) {\sl year}edge[-] (n11);
		\node at (2.3, -2) (m113) {\sl title}edge[-] (n11);
		\node at (1,-3.2) {\scriptsize``$\mathit{L.\,Valiant}$''} edge[-](m111);
		\node at (1.6,-2.8) {\scriptsize``$\mathit{1984}$''} edge[-] (m112);
		\node at (2.3,-3.2) {\scriptsize``$\mathit{A\,theory}$} edge[-] (m113);
		\node at (2.3,-3.6) {\scriptsize $\mathit{of\,the\,learnable}$''};
  \end{tikzpicture}
  \caption{\label{fig:dblp}A trivialized DBLP repository.}
\end{figure}

Typically, a \emph{schema} for XML defines for every node its \emph{content model} i.e., the children nodes it must, may, and cannot contain. 
For instance, in the DBLP example, one would require every \rnew{article} to have exactly one \rnew{title}, one \rnew{year}, and one or more \rnew{authors}. 
A \rnew{book} may additionally contain one \rnew{publisher} and may also have one or more \rnew{editors} instead of \rnew{authors}. 
A schema has numerous important uses. 
For instance, it allows to validate a document against a schema and identify potential errors. 
A schema also serves as a reference for a user who does not know yet the structure of the XML document and attempts to query or modify its \rnew{content}.

The \emph{Document Type Definition} (DTD), the most widespread XML schema formalism for (ordered) XML \cite{BeNeVa04,GrMa11}, is
essentially a set of rules associating with each label a regular
expression that defines the admissible sequences of children. 
The DTDs are best fitted for ordered content because they use regular expressions, a formalism that defines sequences of labels. 
However, when unordered content model needs to be defined, there is a tendency to use \emph{over-permissive} regular expressions.
For instance, the DTD below corresponds to the one used in practice for the DBLP repository\rnew{\footnote{\rnew{\url{http://dblp.uni-trier.de/xml/dblp.dtd}}}}:
\begin{flalign*}
\textsl{dblp} \rightarrow &~~ (\textsl{article} \mid \textsl{book})^*\\
\textsl{article} \rightarrow &~~ (\textsl{title} \mid \textsl{year} \mid \textsl{author})^* \\
\textsl{book} \rightarrow &~~ (\textsl{title} \mid \textsl{year} \mid \textsl{author} \mid \textsl{editor} \mid \textsl{publisher})^*
\end{flalign*}
This DTD allows an \rnew{article} to contain any number of \rnew{titles}, \rnew{years}, and \rnew{authors}. 
A \rnew{book} may also have any number of \rnew{titles}, \rnew{years}, \rnew{authors}, \rnew{editors}, and \rnew{publishers}. These regular expressions are clearly over-permissive because they allow XML documents that do not follow the intuitive guidelines set out earlier e.g., an XML document containing an \rnew{article} with two \rnew{titles} and no \rnew{author}
should not be valid.

While it is possible to capture unordered content models with regular expressions, a simple pumping argument shows that their size may need to be exponential in the number of possible labels of the children. 
In case of the DBLP repository, this number reaches values up to 12, which basically precludes any practical use of such regular expressions. 
This suggests that over-permissive regular expressions may be employed for the reasons of conciseness and readability, a consideration of great practical importance.

The use of over-permissive regular expressions, apart from allowing documents that do not follow the guidelines, has other negative consequences e.g., in static analysis tasks that involve the schema. 
Take for example the following two twig
queries~\cite{ACLS02,XPath1}:
\begin{align*}
&/\textsl{dblp}/\textsl{book}[\textsl{author}=\mbox{``$\mathit{C.\,Papadimitriou}$''}]\\
&/\textsl{dblp}/\textsl{book}[\textsl{author}=\mbox{``$\mathit{C.\,Papadimitriou}$''}][\textsl{title}]
\end{align*}
The first query selects the elements labeled \rnew{book}, children of \rnew{dblp} and having an \rnew{author} containing the text ``\emph{C.\ Papadimitriou.}'' 
The second query additionally requires that \rnew{book} has a \rnew{title}. 
Naturally, these two queries should be equivalent because every \rnew{book} should have a \rnew{title}. 
However, the DTD above does not capture properly this requirement, and consequently the two queries are not equivalent w.r.t.\ this DTD.

\rnew{In this paper, we investigate schema languages for unordered XML.
First, we study languages of \emph{unordered words}, where an unordered word can be seen as a multiset of symbols.
We consider \emph{unordered regular expressions} (UREs), which are essentially regular expressions with \emph{unordered concatenation} ``$\shuffle$'' instead of standard concatenation.
The unordered concatenation can be seen as union of multisets, and consequently, the star ``$*$'' can be seen as the Kleene closure of unordered languages.
}
Similarly to a DTD which associates to each label a regular expression to define its (ordered) content model, an unordered schema uses UREs to define for each label its unordered content model. 
For instance, take the following schema (satisfied by the tree in Figure~\ref{fig:dblp}):
\begin{flalign*}
\textsl{dblp} \rightarrow &~~ \textsl{article}^* \shuffle \textsl{book}^*\\
\textsl{article} \rightarrow &~~ \textsl{title} \shuffle \textsl{year} \shuffle \textsl{author}^+\\
\textsl{book} \rightarrow &~~ \textsl{title} \shuffle \textsl{year} \shuffle \textsl{publisher}^? \shuffle (\textsl{author}^+\mid\textsl{editor}^+)
\end{flalign*}
The above schema uses UREs and captures the intuitive requirements for the DBLP repository. 
In particular, an \rnew{article} must have exactly
one \rnew{title}, exactly one \rnew{year}, and at least one \rnew{author}. 
A \rnew{book} may additionally have a \rnew{publisher} and
may have one or more \rnew{editors} instead of \rnew{authors}.  Note that, unlike the DTD defined earlier, this schema does not allow documents having an \rnew{article} with several \rnew{titles} or without any \rnew{author}.

Using UREs is equivalent to using DTDs with regular expressions
interpreted under the \emph{commutative closure}~\cite{BeMi99,NeSc99}: essentially, a word matches the commutative closure of a regular expression if there exists a permutation of the word that matches the regular expression in the standard way.
Deciding this problem is known to be NP-complete~\cite{KoTo10} for arbitrary regular expressions. 
We show that the problem of testing the membership of an unordered word to the language of a URE is NP-complete even for a restricted subclass of UREs that allows unordered concatenation and the option operator ``$?$'' only. 
Not surprisingly, testing the containment of two UREs is also intractable. 
These results are of particular interest because they are novel and do not follow from complexity results for regular expressions, where the order plays typically an essential role~\cite{StMe73,MaSt94}. 
Consequently, we focus on finding restrictions rendering UREs tractable and capable of capturing practical languages in a simple and concise manner.

The first restriction is to disallow repetitions of a symbol in a URE, thus banning expressions of the form $a\shuffle a^?$ because the symbol $a$ is used twice. 
Instead we add general interval multiplicities $a^{[1,2]}$ which offer a way to specify a range of occurrences of a symbol in an unordered word without repeating a symbol in the URE. 
While the complexity of the membership of an unordered word to the language of a URE with interval multiplicities and without symbol repetitions has recently been shown to be in PTIME~\cite{BGHPSS14}, testing containment of two such UREs remains intractable. 
We, therefore, add limitations on the nesting of the disjunction and the unordered concatenation operators and the use of intervals, which yields the proposed class of \emph{disjunctive interval multiplicity expressions} (DIMEs).
DIMEs enjoy good computational properties: both the membership and the containment problems become tractable. 
Also, we believe that despite the imposed restriction DIMEs remain a practical class of UREs. 
For instance, all UREs used in the schema for the DBLP repository above are DIMEs.

\begin{table*}
{\scriptsize
\setlength{\tabcolsep}{.48em}
\begin{tabular}{|l|l|l|l|l|}
  \hline
  \emph{Problem of interest} & $\DTD$ & $\dims$& \em disj.-free $\DTD$ & $\sims$\\\hline
  Schema satisfiability & PTIME \cite{BuWo98,Schwentick04b} & PTIME (Pr. \ref{dims:ptime}) & PTIME \cite{BuWo98,Schwentick04b} & PTIME (Pr. \ref{dims:ptime})  \\\hline
  Membership & PTIME \cite{BuWo98,Schwentick04b} & PTIME (Pr. \ref{dims:ptime2}) & PTIME \cite{BuWo98,Schwentick04b} & PTIME (Pr. \ref{dims:ptime2}) \\\hline
  Schema containment& 
\begin{minipage}{1.9cm}
PSPACE-c${}^\dag$\cite{Schwentick04b} PTIME~\cite{BuWo98}
\end{minipage} 
& PTIME (Pr. \ref{dims:ptime}) & 
\begin{minipage}{1.9cm}
coNP-h${}^\dag$\cite{MaNeSc09} PTIME~\cite{BuWo98}
\end{minipage}
& PTIME (Pr. \ref{dims:ptime}) \\\hline
Query satisfiability${}^\ddag$ & NP-c \cite{BeFaGe08} & NP-c (Pr. \ref{propsatqnphard}) & PTIME \cite{BeFaGe08} & PTIME (Th. \ref{sat-impl-ptime}) \\\hline
Query implication${}^\ddag$ & EXPTIME-c \cite{NeSc06}& EXPTIME-c (Pr. \ref{propexptimec})&PTIME (Th. \ref{thcor})& PTIME (Th. \ref{sat-impl-ptime})\\\hline
Query containment${}^\ddag$ & EXPTIME-c \cite{NeSc06} &EXPTIME-c (Pr. \ref{propexptimec}) &coNP-c (Th. \ref{thcor})&coNP-c (Th. \ref{thconp})  \\\hline
\multicolumn{5}{l}{{${}^\dag$ when non-deterministic regular expressions are used. ${}^\ddag$ for twig queries.}}
\end{tabular}
}
\caption{\label{tab:complexity_summary}Summary of complexity results.}

\end{table*}
Next, we employ DIMEs to define languages of unordered trees and propose two schema languages: \emph{disjunctive interval multiplicity   schema} (DIMS), and its restriction, \emph{disjunction-free interval multiplicity schema} (IMS). 
Naturally, the above schema for the DBLP repository is a DIMS. We study the complexity of several basic decision problems: schema satisfiability, membership of a tree to the language of a schema, containment of two schemas, twig query satisfiability, implication, and containment in the presence of schema. 
We present in Table~\ref{tab:complexity_summary} a summary of the complexity results and we observe that DIMSs and IMSs enjoy the
same computational properties as general DTDs and disjunction-free DTDs, respectively.  

The lower bounds for the decision problems for DIMSs and IMSs are generally obtained with easy adaptations of their counterparts for
general DTDs and disjunction-free DTDs. 
To obtain the upper bounds we develop several new tools. 
We propose to represent DIMEs with \emph{characterizing tuples} that can be efficiently computed and allow deciding in polynomial time the membership of a tree to the language of a DIMS and the containment of two DIMSs.
Also, we develop \emph{dependency graphs} for IMSs and a generalized definition of an \emph{embedding} of a query. 
These two tools help us to reason about query satisfiability, query implication, and query containment in the presence of IMSs. 
Our constructions and results for IMSs allow also to characterize the complexity of query implication and query containment in the presence of disjunction-free DTDs, which, to the best of our knowledge, have not been previously studied.

Finally, we compare the expressive power of the proposed schema languages with yardstick languages of unordered trees (FO, MSO, and
Presburger constraints) and DTDs under commutative closure.
We show that the proposed schema languages are capable of expressing many practical languages of unordered trees.

It is important to mention that this paper is a substantially extended version of a preliminary work presented in~\cite{BoCiSt13}.
More precisely, in this paper we show novel intractability results for some subclasses of unordered regular expressions and we extend the expressibility of the tractable subclasses.
While in~\cite{BoCiSt13} we have considered only simple multiplicities ($*,+,?$), in this paper we deal with arbitrary interval multiplicities of the form $[n,m]$.

{\bf Organization.}
In Section~\ref{sec:prelim} we introduce some preliminary notions.
In Section~\ref{sec:ure} we study the reasons of intractability of unordered regular expressions while in Section~\ref{subsec:dime} we present the tractable subclass of \emph{disjunctive interval multiplicity expressions} (DIMEs).
In Section~\ref{sec:static} we define two schema languages: the \emph{disjunctive interval multiplicity schemas} (DIMSs) and its restriction, the \emph{disjunction-free interval multiplicity schemas} (IMSs), and the related problems of interest.
In Section~\ref{dims} and Section~\ref{ms} we analyze the complexity of the problems of interest for DIMSs and IMSs, respectively.
In Section~\ref{sec:expressiveness} we discuss the expressiveness of the proposed formalisms.
\rnew{In Section~\ref{sec:related:work} we present related work.}
In Section~\ref{sec:conclusions} we summarize our results and outline further directions.

\section{Preliminaries}\label{sec:prelim}
Throughout this paper we assume an alphabet $\Sigma$ that is a finite set of symbols.
We also assume that $\Sigma$ has a total order $<_\Sigma$ that can be tested in constant time.

{\bf Trees.} 
We model XML documents with unordered labeled trees.
Formally, a {\em tree} $t$ is a tuple $(N_t,\root_t,\lab_t$, $\child_t)$, where $N_t$ is a finite set of nodes, $\root_t\in N_t$ is a distinguished root node, $\lab_t:N_t\rightarrow\Sigma$ is a labeling function, and $\child_t\subseteq N_t\times N_t$ is the parent-child relation.
We assume that the relation $\child_t$ is acyclic and require every non-root node to have exactly one predecessor in this relation.
By $\Tree$ we denote the set of all trees.

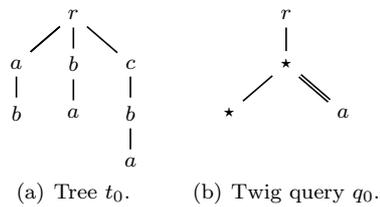
\begin{figure}[htb]
  \centering
  \subfigure[Tree $t_0$.]{\label{fig:tree}
    \centering
    \begin{tikzpicture}[yscale=0.65]
      \path[use as bounding box] (-1.25,.25) rectangle (1.25,-3.25);
      \node at (0,0) (n0) {$r$};
      \node at (-0.75,-1) (n1) {$a$};
      \node at (0,-1) (n4) {$b$};
      \node at (0,-2) (n5) {$a$};
      \node at (0.75,-1) (n2) {$c$};
      \node at (0.75,-2) (n3) {$b$};
      \node at (0.75,-3) (n6) {$a$};
      \node at (-0.75,-2) (n7) {$b$};
      \draw[-,semithick] (n1) -- (n7);      
      \draw[-,semithick] (n0) -- (n1);
      \draw[-,semithick] (n0) -- (n1);
      \draw[-,semithick] (n0) -- (n2);
      \draw[-,semithick] (n2) -- (n3);
      \draw[-,semithick] (n3) -- (n6);
      \draw[-,semithick] (n0) -- (n4);
      \draw[-,semithick] (n4) -- (n5);
    \end{tikzpicture}
  }
  \subfigure[Twig query $q_0$.]{\label{fig:twig0}
    \centering
    \begin{tikzpicture}[yscale=0.65]
      \path[use as bounding box] (-1.25,.25) rectangle (1.25,-3.25);
      \node at (0,0) (n0) {$r$};
      \node at (0,-1) (n2) {$\wc$};
      \node at (0.75,-2) (n3) {$a$};
      \node at (-0.75,-2) (n1) {$\wc$};
      \draw[-,semithick] (n2) -- (n1);
      \draw[-,semithick] (n0) -- (n2);
      \draw[-,double,semithick] (n2) -- (n3);
    \end{tikzpicture}
  }
  \caption{A tree and a twig query.}
  \label{fig:trees-twigs}
\end{figure}
{\bf Queries.}  
We work with the class of twig queries, which are essentially unordered trees whose nodes may be additionally labeled with a distinguished wildcard symbol $\wc\not\in\Sigma$ and that use two types of edges, child ($/$) and descendant ($\dblslash$),
corresponding to the standard XPath axes. 
Note that the \rnew{semantics of the $\dblslash$-edge} is that of a proper descendant (and not that of descendant-or-self).
Formally, a \emph{twig query} $q$ is a tuple $(N_q,\root_q,\lab_q,\child_q,$ $\desc_q)$, where $N_q$ is a finite set of nodes, $\root_q\in N_q$ is the root node, $\lab_q:N_q\rightarrow\Sigma\cup\{\wc\}$ is a labeling function, $\child_q\subseteq N_q\times N_q$ is a set of child edges, and $\desc_q\subseteq N_q\times N_q$ is a set of descendant edges.  
We assume that $\child_q\cap\desc_q=\emptyset$ and that the relation $\child_q\cup\desc_q$ is acyclic and we require every non-root node to have exactly one predecessor in this relation. 
By $\Twig$ we denote the set of all twig queries. 
Twig queries are often presented using the abbreviated XPath syntax~\cite{XPath1} e.g., the query $q_0$ in Figure~\ref{fig:twig0} can be written as $r/\wc[\wc]\dblslash{}a$.

{\bf Embeddings.} 
We define the semantics of twig queries using the notion of embedding which is essentially a mapping of nodes of a query to the nodes of a tree that respects the semantics of the edges of the query. 
Formally, for a query $q\in \Twig$ and a tree $t\in \Tree$, an \emph{embedding} of $q$ in $t$ is a function $\lambda : N_q \rightarrow N_t$ such that:
\begin{enumerate}
\itemsep0pt
\item[$1$.] $\lambda(\root_q)=\root_t$,
\item[$2$.] for every $(n,n')\in \child_q$,
  $(\lambda(n),\lambda(n'))\in \child_t$,
\item[$3$.] for every $(n,n')\in \desc_q$,
  $(\lambda(n),\lambda(n'))\in( \child_t)^+$ (the transitive closure of $\child_t$),
\item[$4$.] for every $n\in N_q$, $\lab_q(n) = \wc$ or $\lab_q(n) = \lab_t(\lambda(n))$.
\end{enumerate}
\rnew{We write $t\preccurlyeq q$ if there exists an embedding of $q$ in $t$.
Later on, in Section~\ref{subsec:embed} we generalize this definition of embedding as a tool that permits us characterizing the problems of interest.

As already mentioned, we use the notion of embedding to define the semantics of twig queries.
In particular, we say that $t$ \emph{satisfies} $q$ if there exists an embedding of $q$ in $t$ and we write $t\models q$.
By $ L(q)$ we denote the set of all trees satisfying $q$. 
}

Note that we do not require the embedding to be injective i.e., two nodes of the query may be mapped to the same node of the tree.  
Figure~\ref{fig:embeddings} presents all embeddings of the query $q_0$ in the tree $t_0$ from Figure~\ref{fig:trees-twigs}.

\begin{figure}[htb]
  \centering
  \begin{tikzpicture}[yscale=0.65]
    \node at (0,0) (n0) {$r$};
    \node at (-0.75,-1) (n1) {$a$};
    \node at (0,-1) (n4) {$b$};
    \node at (0,-2) (n5) {$a$};
    \node at (0.75,-1) (n2) {$c$};
    \node at (0.75,-2) (n3) {$b$};
    \node at (0.75,-3) (n6) {$a$};
    \node at (-0.75,-2) (n7) {$b$};
    \draw[-,semithick] (n1) -- (n7);      
    \draw[-,semithick] (n0) -- (n1);
    \draw[-,semithick] (n0) -- (n2);
    \draw[-,semithick] (n2) -- (n3);
    \draw[-,semithick] (n3) -- (n6);
    \draw[-,semithick] (n0) -- (n4);
    \draw[-,semithick] (n4) -- (n5);
    \begin{scope}[xshift=-2.75cm]
      \node at (0,0) (m0) {$r$};
      \node at (0,-1) (m2) {$\wc$};
      \node at (0.75,-2) (m3) {$a$};
      \node at (-0.75,-2) (m1) {$\wc$};
      \draw[-,semithick] (m2) -- (m1);
      \draw[-,semithick] (m0) -- (m2);
      \draw[-,double,semithick] (m2) -- (m3);
    \draw (m0) edge[->,bend left,densely dotted] (n0);
    \draw (m1) edge[->,bend right,densely dotted] (n5);
    \draw (m2) edge[->,bend left,densely dotted] (n4);
    \draw (m3) edge[->,bend left,densely dotted] (n5);
    \end{scope}
    \begin{scope}[xshift=2.75cm]
      \node at (0,0) (m0) {$r$};
      \node at (0,-1) (m2) {$\wc$};
      \node at (0.75,-2) (m3) {$a$};
      \node at (-0.75,-2) (m1) {$\wc$};
      \draw[-,semithick] (m2) -- (m1);
      \draw[-,semithick] (m0) -- (m2);
      \draw[-,double,semithick] (m2) -- (m3);
    \draw (m0) edge[->,bend right,densely dotted] (n0);
    \draw (m1) edge[->,bend right,densely dotted] (n3);
    \draw (m2) edge[->,bend right,densely dotted] (n2);
    \draw (m3) edge[->,bend left,densely dotted] (n6);
    \end{scope}
  \end{tikzpicture}
  \caption{Embeddings of $q_0$ in $t_0$.}
  \label{fig:embeddings}
\end{figure}
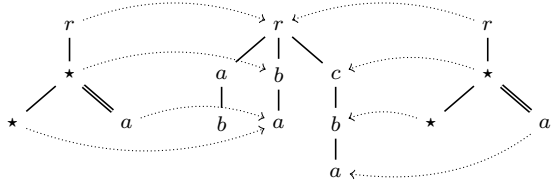
{\bf Unordered words.} 
An \emph{unordered word} is essentially a multiset of symbols i.e., a function $w:\Sigma\rightarrow\mathbb N_0$ mapping symbols from the alphabet to natural numbers. 
We call $w(a)$ the number of occurrences of the symbol $a$ in $w$. 
We also write $a\in w$ as a shorthand for $w(a)\neq 0$. 
An empty word $\varepsilon$ is an unordered word that has $0$ occurrences of every symbol i.e., $\varepsilon(a)=0$ for every
$a\in\Sigma$. 
We often use a simple representation of unordered words, writing each symbol in the alphabet the number of times it occurs in
the unordered word. For example, when the alphabet is $\Sigma=\{a,b,c\}$, $w_0=aaacc$ stands for the function $w_0(a) = 3$,
$w_0(b) = 0$, and $w_0(c) = 2$. 
Additionally, we may write $w_0=a^3c^2$ instead of $w_0=aaacc$.

We use unordered words to model collections of children of XML nodes. 
\rnew{
As it is usually done in the context of XML validation~\cite{SeVi02,SeSi07}, we assume that the XML document is encoded in unary i.e., every node takes the same amount of memory.
Thus, we use a unary representation of unordered words, where each occurrence of a symbol occupies the same amount of space.
However, we point out that none of the results presented in this paper changes with a binary representation.
In particular, the intractability of the membership of an unordered word to the language of a URE (Theorem~\ref{th:memb:npc}) also holds with a binary representation of unordered words.
}

Consequently, the \emph{size} of an unordered word $w$, denoted $|w|$, is the sum of the numbers of occurrences in $w$ of all symbols in the alphabet. 
For instance, the size of $w_0=aaacc$ is $|w_0|=5$.

The (\emph{unordered}) \emph{concatenation} of two unordered words $w_1$ and $w_2$ is defined as the multiset union $w_1\uplus w_2$ i.e., the function defined as $(w_1\uplus w_2)(a) = w_1(a)+w_2(a)$ for every $a\in\Sigma$. 
For instance, $aaacc\uplus{}abbc=aaaabbccc$. 
Note that $\varepsilon$ is the identity element of the unordered concatenation $\varepsilon\uplus w = w\uplus \varepsilon = w$ \rnew{for every unordered word} $w$. 
Also, given an unordered word $w$, by $w^i$ we denote the concatenation $w\uplus\ldots\uplus w$ ($i$ times).

A \emph{language} is a set of unordered words. 
The unordered concatenation of two languages $L_1$ and $L_2$ is a language $L_1\uplus L_2 = \{w_1\uplus w_2\mid w_1\in L_1, w_2\in L_2\}$. 
For instance, if $L_1 = \{a, aac\}$ and $L_2 = \{ac, b, \varepsilon\}$, then $L_1\uplus L_2 = \{a,ab,aac,aabc,aaacc\}$.

{\bf Unordered regular expressions.} 
Analogously to regular expressions, which are used to define languages of ordered words, we propose unordered regular expressions to define languages of unordered words.
Essentially, an \emph{unordered regular expression} (URE) defines unordered words by using Kleene star ``$*$'', disjunction ``$\mid$'', and unordered concatenation ``$\shuffle$''.
Formally, we have the following grammar:
\[
E ::= \epsilon \mid a  \mid E^* \mid (E\mbox{``$\mid$''} E) \mid (E\mbox{``$\shuffle$''} E),
\]
where $a\in\Sigma$.  
The semantics of UREs is defined as follows:
\begin{align*}
&L(\epsilon) = \{ \varepsilon \}, \\
&L(a) = \{a\}, \\
&L(E_1\mid E_2) = L(E_1) \cup L(E_2),\\
&L(E_1\shuffle E_2) = L(E_1) \uplus L(E_2),\\
&L(E^*) = \{w_1\uplus\ldots\uplus w_i\mid w_1,\ldots,w_i\in L(E)\wedge i\geq0\}.
\end{align*}
For instance, the URE $(a\shuffle (b\mid c))^*$ accepts the unordered words having the number of occurrences of $a$ equal to the total number of $b$'s and $c$'s.

The grammar above uses only one \emph{multiplicity} $*$ and we introduce macros for two other standard and commonly used multiplicities:
\begin{gather*}
E^+ \colonequals E\shuffle E^*, \qquad 
E^? \colonequals E\mid\epsilon.
\end{gather*}
The URE $(a\shuffle b^?)^+ \shuffle (a\mid c)^?$ accepts the unordered words having at least one $a$, at most one $c$, and a number of $b$'s less or equal than the number of $a$'s.

{\bf Interval multiplicities.} 
While the multiplicities $*$, $+$, and $?$ allow to specify unordered words with multiple occurrences of a symbol, we additionally introduce \emph{interval multiplicities} to allow to specify a range of allowed occurrences of a symbol in an unordered word. 
More precisely, we extend the grammar of UREs by allowing expressions of the form $E^{[n,m]}$ and \rnew{$E^{[n,m]^?}$}, where $n\in\mathbb N_0$ and $m\in \mathbb N_0\cup\{\infty\}$. 
Their semantics is defined as follows:
\begin{align*}  
  &L(E^{[n,m]}) = \{w_1\uplus\ldots\uplus w_i\mid {}
  w_1,\ldots,w_i\in L(E) \land n \leq i\leq m\},
  \\
  & \rnew{L(E^{[n,m]^?})} = L(E^{[n,m]}) \cup \{\varepsilon\}.
\end{align*}
In the rest of the paper, we write simply \emph{interval} instead of \emph{interval multiplicity}. 
Furthermore, we view the following standard multiplicities as macros for intervals:
\begin{gather*}
* := [0,\infty], \qquad + := [1,\infty], \qquad ? := [0,1].
\end{gather*}
Additionally, we introduce the single occurrence multiplicity $1$ as a macro for the interval $[1,1]$.

Note that the intervals do not add expressibility to general UREs, but they become useful if we impose some restrictions.
For example, if we disallow repetitions of a symbol in a URE and ban expressions of the form $a\shuffle a^?$, we can however write $a^{[1,2]}$ to specify a range of occurrences of a symbol in an unordered word without repeating a symbol in the URE.

\section{Intractability of unordered regular expressions}\label{sec:ure}
In this section, we study the reasons of the intractability of UREs w.r.t.\ the following two fundamental decision problems: \emph{membership} and \emph{containment}.
In Section~\ref{subsec:memb} we show that membership is NP-complete even under significant restrictions on the UREs while in Section~\ref{subsec:cnt} we show that the containment is $\pip$-hard (and in 3-EXPTIME).
We notice that the proofs of both results rely on UREs allowing repetitions of the same symbol.
Consequently, we disallow such repetitions and we show that this restriction does not avoid intractability of the containment (Section~\ref{subsec:interval}).
We observe that the proof of this result employs UREs with arbitrary use of disjunction and intervals, and therefore, in Section~\ref{subsec:dime} we impose further restrictions and define the disjunctive interval multiplicity expressions (DIMEs), a subclass for which we show that the two problems of interest become tractable.

\subsection{Membership}\label{subsec:memb}
In this section, we study the problem of deciding the \emph{membership} of an unordered word to the language of a URE.
First of all, note that this problem can be easily reduced to testing the membership of a vector to the Parikh image of a regular language, known to be NP-complete~\cite{KoTo10}, and vice versa. 
We show that deciding the membership of an unordered word to the language a URE remains NP-complete even under significant restrictions on the class of UREs, a result which does not follow from~\cite{KoTo10}.

\begin{theorem}\label{th:memb:npc}
Given an unordered word $w$ and an expression $E$ of the grammar $E ::= a\mid E^?\mid (E\text{\emph{``$\shuffle$''}} E)$, deciding whether $w\in L(E)$ is NP-complete. 
\end{theorem}
\begin{proof}
To show that this problem is in NP, we point out that a nondeterministic Turing machine guesses a permutation of $w$ and checks whether it is accepted by the NFA corresponding to $E$ with the unordered concatenation replaced by standard concatenation. 
We recall that $w$ has unary representation.

Next, we prove the NP-hardness by reduction from $\satt$ i.e., given a 3CNF formula, determine whether there exists a valuation such that each clause has exactly one true literal (and exactly two false literals).
The $\satt$ \rnew{problem} is known to be NP-complete~\cite{Schaefer78}.
The reduction works as follows.
We take a 3CNF formula $\varphi=c_1\wedge\ldots\wedge c_k$ over the variables $\{x_1,\ldots,x_n\}$.
We take the alphabet $\{d_1,\ldots,d_k,v_1,\ldots,v_n\}$. Each $d_i$ corresponds to a clause $c_i$ (for $1\leq i\leq k$) and each $v_j$ corresponds to a variable $x_j$ (for $1\leq j\leq n$).
We construct the unordered word $w_\varphi=d_1\ldots d_kv_1\ldots v_n$ and the expression $E_\varphi = X_1\shuffle\ldots\shuffle X_n$, where for $1\leq j\leq n$:
\[
 X_j = (v_j\shuffle d_{t_1}\shuffle\ldots\shuffle d_{t_l})^?\shuffle(v_j\shuffle d_{f_1}\shuffle \ldots\shuffle d_{f_m})^?,
\]
and $d_{t_1},\ldots,d_{t_l}$ (with $1\leq t_1,\ldots,t_l\leq k$) correspond to the clauses that use the literal $x_j$, and $d_{f_1},\ldots d_{f_m}$  (with $1\leq f_1,\ldots,f_m\leq k$) correspond to the clauses that use the literal $\neg x_j$.
For example, for the formula $\varphi_0=(x_1\vee \neg x_2\vee x_3)\wedge (\neg x_1\vee x_3\vee \neg x_4)$, we construct $w_{\varphi_0}=d_1d_2v_1v_2v_3v_4$ and
\[
 E_{\varphi_0}=(v_1\shuffle d_1)^?\shuffle (v_1\shuffle d_2)^?\shuffle v_2^?\shuffle (v_2\shuffle d_1)^?\shuffle(v_3\shuffle d_1\shuffle d_2)^?\shuffle v_3^?\shuffle v_4^?\shuffle (v_4\shuffle d_2)^?.
\]
We claim that $\varphi\in\satt$ iff $w_\varphi\in L(E_\varphi)$.
For the \emph{only if} case, let $V:\{x_1,\ldots,x_n\}\rightarrow\{\true,\false\}$ be the $\satt$ valuation of $\varphi$.
We use $V$ to construct the derivation of $w_\varphi$ in $L(E_\varphi)$: for $1\leq j\leq n$, we take $(v_j\shuffle d_{t_1}\shuffle\ldots\shuffle d_{t_l})$ from $X_j$ if $V(x_j)=\true$, and $(v_j\shuffle d_{f_1}\shuffle \ldots\shuffle d_{f_m})$ from $X_j$ otherwise.
Since $V$ is a $\satt$ valuation of $\varphi$, each $d_i$ (with $1\leq i\leq k$) occurs exactly once, hence $w_\varphi\in L(E_\varphi)$.
For the \emph{if} case, we assume that $w_\varphi\in L(E_\varphi)$.
Since $w_\varphi(v_j) = 1$, we infer that $w_\varphi$ uses exactly one of the expressions of the form $(v_j\shuffle\ldots)^?$.
Moreover, since $w_\varphi(d_i) = 1$, we infer that the valuation encoded in the derivation of $w_\varphi$ in $L(E_\varphi)$ validates exactly one literal of each clause in $\varphi$, and therefore, $\varphi\in\satt$.
Clearly, the described reduction works in polynomial time.
\qed\end{proof}

\subsection{Containment}\label{subsec:cnt}
In this section, we study the problem of deciding the \emph{containment} of two UREs.
It is well known that regular expression containment is a PSPACE-complete problem~\cite{StMe73}, but we cannot adapt this result to characterize the complexity of the containment of UREs because the order plays an essential role in the reduction.
In this section, we prove that deciding the containment of UREs is $\pip$-hard and we show an upper bound which follows from the
complexity of deciding the satisfiability of Presburger logic formulas~\cite{Oppen78,SeScMu08}.

\begin{theorem}\label{th:pip}
Given two UREs $E_1$ and $E_2$, deciding $L(E_1)\subseteq L(E_2)$ is \emph{1)} $\pip$-hard and \emph{2)} in 3-EXPTIME.
\end{theorem}
\begin{proof}
1) We prove the $\pip$-hardness by reduction from the problem of checking the satisfiability of $\qbf$ formulas, a classical $\pip$-complete problem. 
We take a $\qbf$ formula 
\[
 \psi = \forall x_1,\ldots,x_n.\ \exists y_1,\ldots,y_m.\ \varphi,
\]
where $\varphi=c_1\wedge\ldots\wedge c_k$ is a quantifier-free CNF formula. 
We call the variables $x_1,\ldots,x_n$ \emph{universal} and the variables $y_1,\ldots,y_m$ \emph{existential}.

We take the alphabet $\{d_1,\ldots,d_k,t_1,f_1,\ldots,t_n,f_n\}$ and we construct two expressions, $E_\psi$ and $E_\psi'$. 
First, $E_\psi = d_1\shuffle\ldots\shuffle d_k\shuffle X_1\shuffle\ldots\shuffle X_n$, where for $1\leq i\leq n$ $X_i = ((t_i\shuffle d_{a_1}\shuffle\ldots\shuffle d_{a_l})\mid(f_i\shuffle d_{b_1}\shuffle\ldots\shuffle d_{b_s}))$, and $d_{a_1},\ldots d_{a_l}$ (with $1\leq a_1,\ldots,a_l\leq k$) correspond to the clauses which use the literal $x_i$, and $d_{b_1},\ldots,d_{b_s}$ (with $1\leq b_1,\ldots,b_s\leq k$) correspond to the clauses which use the literal $\neg x_i$.
For example, for the formula 
\[
 \psi_0 = \forall x_1,x_2.\ \exists y_1,y_2.\ (x_1\vee\neg x_2\vee y_1)\wedge (\neg x_1\vee y_1\vee\neg y_2)\wedge (x_2\vee\neg y_1),
\]
we construct:
\[
 E_{\psi_0} = d_1\shuffle d_2\shuffle d_3\shuffle ((t_1\shuffle d_1)\mid(f_1\shuffle d_2) )\shuffle ((t_2\shuffle d_3)\mid(f_2\shuffle d_1)).
\]
Note that there is an one-to-one correspondence between the unordered words in $L(E_\psi)$ and the valuations of the universal  variables. 
For example, given the formula $\psi_0$, the unordered word $d_1^3d_2d_3t_1f_2$ corresponds to the valuation $V$ such that $V(x_1) = \true$ and $V(x_2)=\false$.

Next, we construct $E_\psi' = X_1\shuffle\ldots\shuffle X_n\shuffle Y_1\shuffle\ldots\shuffle Y_m$, where:
\begin{itemize}
\item $X_i = ((t_i\shuffle d_{a_1}^*\shuffle\ldots\shuffle d_{a_l}^*)\mid(f_i\shuffle d_{b_1}^*\shuffle\ldots\shuffle d_{b_s}^*))$, and $d_{a_1},\ldots d_{a_l}$ (with $1\leq a_1,\ldots,a_l\leq k$) correspond to the clauses which use the literal $x_i$, and $d_{b_1},\ldots,d_{b_s}$ (with $1\leq b_1,\ldots,b_s\leq k$) correspond to the clauses which use the literal $\neg x_i$ (for $1\leq i\leq n$), 
\item $Y_j = ((d_{a_1}^*\shuffle\ldots\shuffle d_{a_l}^*)\mid(d_{b_1}^*\shuffle\ldots\shuffle d_{b_s}^*))$, and $d_{a_1},\ldots d_{a_l}$ (with $1\leq a_1,\ldots,a_l\leq k$) correspond to the clauses which use the literal $y_j$, and $d_{b_1},\ldots,d_{b_s}$ (with $1\leq b_1,\ldots,b_s\leq k$) correspond to the clauses which use the literal $\neg y_j$ (for $1\leq j\leq m$).
\end{itemize}
For example, for $\psi_0$ above we construct:
\[
 E_{\psi_0}' = ((t_1\shuffle d_1^*)\mid(f_1\shuffle d_2^*))\shuffle ((t_2\shuffle d_3^*)\mid(f_2\shuffle d_1^*))\shuffle ((d_1^*\shuffle d_2^*)\mid d_3^*)\shuffle (\epsilon\mid d_2^*).
\]
We claim that $\models\psi$ iff $E_\psi\subseteq E_\psi'$. 
For the \emph{only if} case, for each valuation of the universal variables, we take the corresponding unordered word $w\in L(E_\psi)$. 
Since there exists a valuation of the existential variables which satisfies $\varphi$, we use this valuation to construct a derivation of $w$ in $L(E_\psi')$. 
For the \emph{if} case, for every unordered word from $L(E_\psi)$, we take its derivation in $L(E_\psi')$ and we use it to construct a valuation of the existential variables which satisfies $\varphi$. 
Clearly, the described reduction works in polynomial time.

2) The membership of the problem to 3-EXPTIME follows from the complexity of deciding the satisfiability of Presburger logic formulas, which is in 3-EXPTIME~\cite{Oppen78}. 
Given two UREs $E_1$ and $E_2$, we compute in linear time~\cite{SeScMu08} two existential Presburger formulas for their Parikh images: $\varphi_{E_1}$ and $\varphi_{E_2}$, respectively.
Next, we test the satisfiability of the following closed Presburger logic formula: $\forall \overline{x}.\ \varphi_{E_1}(\overline{x})\Rightarrow \varphi_{E_2}(\overline{x})$.
\qed\end{proof}
\rnew{While the complexity gap for the containment of UREs (as in Theorem~\ref{th:pip}) is currently quite important, we believe that this gap may be reduced by working on quantifier elimination for the Presburger formula obtained by translating the containment of UREs (as shown in the second part of the proof of Theorem~\ref{th:pip}).
Although we believe that this problem is $\pip$-complete, its exact complexity remains an open question.
}

\subsection{Disallowing repetitions}\label{subsec:interval}
The proofs of Theorem~\ref{th:memb:npc} and Theorem~\ref{th:pip} rely on UREs allowing repetitions of the same symbol, which might be one of the causes of the intractability. 
Consequently, from now on we disallow repetitions of the same symbol in a URE. 
Similar restrictions are commonly used for the regular expressions to maintain practical aspects: \emph{single occurrence regular expressions} (\emph{SOREs})~\cite{BNSV10}, \emph{conflict-free types}~\cite{CGPS13,CoGhSa09,GhCoSa08}, and \emph{duplicate-free DTDs}~\cite{MoWoMo07}.
While the complexity of the membership of an unordered word to the language of a URE without symbol repetitions
has recently been shown to be in PTIME~\cite{BGHPSS14}, testing containment of two such UREs continues to be intractable.
\begin{theorem}\label{th:conphwithoutrepetitions}
Given two UREs $E_1$ and $E_2$ not allowing repetitions of symbols, deciding $L(E_1)\subseteq L(E_2)$ is coNP-hard.
\end{theorem}
\begin{proof}
We show the coNP-hardness by reduction from the complement of 3SAT.
Take a 3CNF formula $\varphi=c_1\wedge\ldots\wedge c_k$ over the variables $\{x_1,\ldots,x_n\}$.  
We assume w.l.o.g.\ that each variable occurs at most once in a clause.	
Take the alphabet $\{a_{ij}\mid 1\leq i\leq k, 1\leq j\leq n,$ $c_i$ uses $x_j$ or $\neg x_j\}$. 
We construct the expression $E_\varphi = X_1\shuffle\ldots\shuffle X_n$, where $X_j=((a_{t_1j}\shuffle\ldots\shuffle a_{t_lj})\mid
(a_{f_1j}\shuffle\ldots\shuffle a_{f_mj})$ (for $1\leq j\leq n$), and $c_{t_1},\ldots,c_{t_l}$ (with $1\leq t_1,\ldots,t_l\leq k$) are the clauses which use the literal $x_j$, and $c_{f_1},\ldots,c_{f_m}$ (with $1\leq f_1,\ldots,f_m\leq k$) are the clauses which use the literal $\neg x_j$. 
Next, we construct $E_\varphi'= (C_1\mid\ldots\mid C_k)^{[0,k-1]}$, where $C_i=(a_{ij_{1}}\mid \ldots\mid a_{ij_{p}})^+$ (for $1\leq i\leq k$), and $x_{j_1},\ldots,x_{j_p}$ (with $1\leq j_1,\ldots,j_p\leq n$) are the variables used by the clause $c_i$. 
For example, for
\[
 \varphi_0=(x_1\vee \neg x_2\vee x_3)\wedge (\neg x_1\vee x_3\vee \neg x_4)\wedge (x_2\vee \neg x_3\vee \neg x_4),
\]
we obtain:
\begin{flalign*}
 E_{\varphi_0} =&(a_{11}\mid a_{21})\shuffle (a_{32}\mid a_{12})\shuffle ((a_{13}\shuffle a_{23})\mid a_{33})\shuffle (\epsilon\mid (a_{24}\shuffle a_{34})),\\
 E_{\varphi_0}' =& ((a_{11}\mid a_{12}\mid a_{13})^+\mid (a_{21}\mid a_{23}\mid a_{24})^+\mid (a_{32}\mid a_{33}\mid a_{34})^+)^{[0,2]}.
\end{flalign*}
Note that there is an one-to-one correspondence between the unordered words $w_V$ in $L(E_\varphi)$ and the valuations $V$ of the  variables $x_1,\ldots,x_n$ (*). 
For example, for above $\varphi_0$ and the valuation $V$ such that $V(x_1)=V(x_2)=V(x_3)=\true$ and $V(x_4)=\false$, the unordered word $w_V=a_{11}a_{32}a_{13}a_{23}a_{24}a_{34}$ \rnew{is in} $L(E_{\varphi_0})$. 
Moreover, given an $w_V\in L(E_{\varphi})$, one can easily obtain the valuation.

We observe that the interval $[0,k-1]$ is used above a disjunction of $k$ expressions of the form $C_i$ and there is no repetition of symbols among the expressions of the form $C_i$.
This allows us to state an instrumental property (**):
$w\in L(E_\varphi')$ iff there exists an $i\in\{1,\ldots,k\}$ such that none of the symbols used in $C_i$ occurs in $w$.
From (*) and (**), we infer that given a valuation $V$, $V\models\varphi$ iff $w_V\in L(E_\varphi)\,\backslash\,L(E_\varphi')$, that yields $\varphi\in$ 3SAT iff $L(E_\varphi)\varnot\subseteq L(E_\varphi')$.
Clearly, the described reduction works in polynomial time.
\qed\end{proof}
Theorem~\ref{th:conphwithoutrepetitions} shows that disallowing repetitions of symbols in a URE does not avoid the intractability of the containment.
\rnew{
Additionally, we observe that the proof of Theorem~\ref{th:conphwithoutrepetitions} employs UREs with arbitrary use of disjunction and intervals.}
Consequently, in the next section we impose further restrictions that yield a class of UREs with desirable computational properties.

\section{Disjunctive interval multiplicity expressions (DIMEs)}\label{subsec:dime}

In this section, we present the DIMEs, a subclass of UREs for which membership and containment become tractable.
\rnew{First, we present an intuitive representation of DIMEs with characterizing tuples (Section~\ref{subsec:characterizing-tuple}).
Next, we formally define DIMEs and show that they are precisely captured by their characterizing tuples (Section~\ref{subsec:definition}).
Finally, we use a compact representation of the characterizing tuples to show the tractability of DIMEs (Section~\ref{subsec:tractability-dime}).}

\subsection{Characterizing tuples}\label{subsec:characterizing-tuple}
In this section, we introduce the notion of \emph{characterizing tuple} that is an alternative, more intuitive representation of DIMEs, the subclass of UREs that we formally define in Section~\ref{subsec:definition}.
Recall that by $a\in w$ we denote $w(a)\neq 0$.
Given a DIME $E$, the \emph{characterizing tuple} $\Delta_E=(C_E,N_E,P_E,K_E)$ is as follows.
\begin{itemize}
\item The \emph{conflicting pairs of siblings} $C_E$ consisting of all pairs of symbols in $\Sigma$ such that $E$ defines no word using both symbols simultaneously: 
\[
 C_E = \{(a,b)\in\Sigma\times\Sigma\mid \varnot\exists w\in L(E).\ a\in w\wedge b\in w\}.
\]
\item The \emph{extended cardinality map} $N_E$ capturing for each symbol in the alphabet the possible numbers of its occurrences in the unordered words defined by $E$:
\[
 N_{E} = \{(a, w(a))\in \Sigma\times\mathbb N_0\mid w\in  L(E)\}.
\]
\item The \emph{collections of required symbols} $P_E$ capturing symbols that must be present in every word; essentially, a set of symbols $X$ belongs to $P_E$ if every word defined by $E$ contains at least one element from $X$:
\[
 P_E=\{X\subseteq\Sigma\mid\forall w\in L(E).\ \exists a\in  X.\ a \in w\}.
\]
\item The \emph{counting dependencies} $K_E$ consisting of pairs of symbols $(a,b)$ such that \rnew{in every word defined by $E$, the number of $b$s is at most the number of $a$s}.
Note that if both $(a,b)$ and $(b,a)$ belong to $K_E$, then all unordered words defined by $E$ should have the same number of $a$'s and $b$'s.
\begin{gather*}
K_E=\{(a,b)\in\Sigma\times\Sigma\mid\forall w\in L(E).\ w(a)\geq w(b)\}.
\end{gather*}
\end{itemize}
As an example we take $E_0= a^+ \shuffle ((b \shuffle c^?)^+\mid d^{[5,\infty]})$ and we illustrate its characterizing tuple $\Delta_{E_0}$.
Because $P_E$ is closed under supersets, we list only its minimal elements:
\begin{align*}
 & C_{E_0} = \{ (b,d), (c,d), (d,b), (d,c) \}, \\
 & N_{E_0}= \{(a,i)\mid i\geq 1\} \cup \{(b,i)\mid i\geq 0\} \cup{}\{(c,i)\mid i\geq 0\} \cup \{(d,i)\mid i = 0 \lor i\geq 5\},\\
 & P_{E_0} =\{\{a\},\{b,d\},\ldots\},\\
 &	K_{E_0} = \{(b,c)\}.
\end{align*}
We point out that $N_E$ may be infinite and $P_E$ exponential in the size of $E$. 
Later on we discuss how to represent both sets in a compact manner while allowing efficient \rnew{manipulation}.

Then, an unordered word $w$ \emph{satisfies} a characterizing tuple $\Delta_E$ corresponding to a DIME $E$, denoted $w\models \Delta_E$, if the following conditions are satisfied: 
\begin{enumerate}
\item $w\models C_E$ i.e., $\forall (a, b) \in C_E.\ (a\in w \Rightarrow b \notin w) \wedge (b\in w \Rightarrow a\notin w)$,
\item $w\models N_E$ i.e., $\forall a \in \Sigma.\ (a,w(a))\in N_E$,
\item $w\models P_E$ i.e., $\forall X\in P_E.\ \exists a\in X.\ a\in w$,
\item $w\models K_E$ i.e., $\forall (a,b)\in K_E.\ w(a)\geq w(b)$.
\end{enumerate}
For instance, the unordered word $\mathit{aabbc}$ satisfies the characterizing tuple $\Delta_{E_0}$ corresponding to the aforementioned DIME $E_0= a^+ \shuffle ((b \shuffle c^?)^+\mid d^{[5,\infty]})$ since it satisfies all the four conditions imposed by $\Delta_{E_0}$. 
On the other hand, note that the following unordered words do not satisfy $\Delta_{E_0}$:
\begin{itemize}
\item $\mathit{abddddd}$ because it contains at the same time $b$ and $d$, and $(b,d)\in C_{E_0}$,
\item $\mathit{add}$ because it has two $d$'s and $(d,2)\notin N_{E_0}$, 
\item $\mathit{aa}$ because it does not contain any $b$ or $d$ and $\{b,d\}\in P_{E_0}$, 
\item $\mathit{abbccc}$ because it has more $c$'s than $b$'s and $(b,c)\in K_{E_0}$.
\end{itemize}
In the next section, we define the DIMEs and show that they are precisely captured by characterizing tuples.

\subsection{Grammar of DIMEs}\label{subsec:definition}
An \emph{atom} is $(a_1^{I_1}\shuffle\ldots\shuffle a_k^{I_k})$, where all $I_i$'s are $?$ or $1$.
For example, $(a\shuffle b^?\shuffle c)$ is an atom, but $(a^{[3,4]} \shuffle b)$ is not an atom.	
A \emph{clause} is $(A_1^{I_1}\mid\ldots\mid A_k^{I_k})$, where all $A_i$'s are atoms and all $I_i$'s are intervals.
A clause is \emph{simple} if all $I_i$'s are $?$ or $1$.
For example, $(a^{[2,3]}\mid (b^?\shuffle c)^*)$ is a clause (which is not simple), $((a^?\shuffle b)\mid c^?)$ is a simple clause while $((a^?\shuffle b^+)\mid c)$ is not a clause.

A \emph{disjunctive interval multiplicity expression} (DIME) is $(D_1^{I_1}\shuffle\ldots\shuffle D_k^{I_k})$, where for $1\leq i\leq k$ either \emph{1)} $D_i$ is a simple clause and $I_i\in\{+,*\}$, or \emph{2)} $D_i$ is a clause and $I_i\in\{1,?\}$.  
Moreover, a symbol can occur at most once in a DIME. 
For example, $(a\mid (b\shuffle c^?)^+)\shuffle (d^{[3,4]}\mid e^*)$ is a DIME while $(a\shuffle b^?)^+\shuffle (a\mid c)$ is not a DIME because it uses the symbol $a$ twice. 
A \emph{disjunction-free interval multiplicity expression} (IME) is a DIME which does not use the disjunction operator. 
An example of IME is $a\shuffle (b\shuffle c^?)^+\shuffle d^{[3,4]}$. 
For more practical examples of DIMEs see Examples~\ref{ex:1} and~\ref{ex:2} from Section~\ref{sec:static}.

We have tailored DIMEs to be able to capture them with characterizing tuples that permit deciding membership and containment in polynomial time (cf.\ Section~\ref{subsec:tractability-dime}).
As we have already pointed out Section~\ref{subsec:interval}, a slightly more relaxed restriction on the nesting of disjunction and intervals leads to intractability of the containment (Theorem~\ref{th:conphwithoutrepetitions}).
Even though DIMEs may look very complex, the imposed restrictions are necessary to obtain lower complexity while considering fragments with practical relevance (cf.\ Section~\ref{sec:expressiveness}).

Next, we show that each DIME can be rewritten as an equivalent \emph{reduced DIME}. 
Reduced DIMEs may also seem complex, but they are a building block for (i) proving that the language of a DIME is precisely captured by its characterizing tuple (Lemma~\ref{lemmawordtuple}), and (ii) computing the compact representation of the characterizing tuples that yield the tractability of DIMEs (cf.\ Section~\ref{subsec:tractability-dime}).

Before defining the reduced DIMEs, we need to introduce some additional notations.
Given an atom $A$ (resp. a clause $D$), we denote by $\Sigma_A$ (resp. $\Sigma_D$) the set of symbols occurring in $A$ (resp. $D$).
Given a DIME $E$, by $I_E^a$ (resp.\ $I_E^A$ or $I_E^D$) we denote the interval associated in $E$ to the symbol $a$ (resp. atom $A$ or clause $D$).
Because we consider only expressions without repetitions, this interval is \rnew{well-defined}.
Moreover, if $E$ is clear from the context, we write simply $I^a$ (resp.\ $I^A$ or $I^D$) instead of $I_E^a$ (resp.\ $I_E^A$ or $I_E^D$).
Furthermore, given an interval $I$  which can be either $[n,m]$ or $[n,m]^?$, by $I^?$ we understand the interval $[n,m]^?$.
In a reduced DIME $E$, each clause with interval $D^I$ has one of the following three types:

\begin{enumerate}
\item $D^I = (A_1\mid\ldots\mid A_k)^+$, where $k\geq 2$ \rnew{and, for every $i\in\{1,\ldots,k\}$, $A_i$ is} an atom such that there exists $a\in\Sigma_{A_i}$ such that $I^a=1$. 

For example, $((a\shuffle b^?)\mid c)^+$ has type 1, but $a^+$ and $((a^?\shuffle b^?)\mid c)^+$ do not.

\item $(A_1^{I_1}\mid\ldots\mid A_k^{I_k})$, where for every $i\in\{1,\ldots, k\}$ \emph{1)} $A_i$ is an atom such that there exists $a\in\Sigma_{A_i}$ such that $I^a=1$ and \emph{2)} $0$ does not belong to the set represented by the interval $I_i$.

For example, $(a\mid (b^?\shuffle c)^{[5,\infty]})$ and $a^+$ have type 2, but $(a\mid (b^?\shuffle c^?)^{[5,\infty]})$ and $(a^*\mid (b^?\shuffle c)^{[5,\infty]})$ do not.

\item $(A_1^{I_1}\mid\ldots\mid A_k^{I_k})$, where for every $i\in\{1,\ldots, k\}$ $A_i$ is an atom and $I_i$ is an interval such that $0$ belongs to the set represented by the interval $I_i$.

For example, $(a^*\mid (b\shuffle c)^{[3,4]^?})$ and $(a^?\shuffle b^?)^*$ have type 3, but $(a^?\shuffle b^?)^{[3,4]}$ does not.
\end{enumerate}
\rnew{
The reduced DIMEs easily yield the construction of their characterizing tuples.
Take a clause with interval $D^{I}$ from a DIME $E$ and observe that the symbols from $\Sigma_{D}$ are present in the characterizing tuple $\Delta_E$ as follows.

\begin{itemize}
\item If $D^{I}$ is of type 1, then there is no symbol in $\Sigma_{D}$ that occurs in a conflict in $C_E$.
Otherwise, $C_E$ consists of all pairs of distinct symbols $(a,b)$ from $\Sigma_{D}$ that appear in different atoms from $D^I$.

\item If $D^{I}$ is of type 1, then we have $(a,n)\in N_E$ for every $(a,n)\in \Sigma_{D}\times\Nz$.
Otherwise, the possible number of occurrences of every symbol $a$ from $\Sigma_D$ can be obtained directly from the two intervals above it: the interval of $D$ and the interval of the atom containing $a$.
We explain in Section~\ref{subsec:tractability-dime} how to precisely construct a compact representation of the potentially infinite set $N_E$.

\item If $D^I$ is of type 1 or 2, then every unordered word defined by $E$ contains at least one of the symbols $a$ from $\Sigma_{D}$ having interval $I^a=1$.
More precisely, $P_E$ contains all sets of symbols $X\subseteq\Sigma$ containing, for every atom of $D$, at least one symbol $a$ with $I^a =1$.
For example, for $((a\shuffle b\shuffle c^?)\mid (d\shuffle e))^+$, $P_E$ consists of the sets $\{a,d\}, \{a,e\}, \{b,d\}, \{b,e\}$ and all their supersets.
Otherwise, if $D^I$ is of type 3, then there is no set in $P_E$ containing only symbols from $\Sigma_D$.

\item Regardless of the type of $D^I$, the counting dependencies $K_E$ consist of all pairs of symbols $(a,b)$ such that they appear in the same atom in $D$ and $I^a=1$.
\end{itemize}
}
To obtain reduced DIMEs, we use the following rules:
\begin{itemize}
\item Take a simple clause $(A_1^{I_1}\mid\ldots\mid A_k^{I_k})$.
\begin{itemize}
\item $(A_1^{I_1}\mid\ldots\mid A_k^{I_k})^*$ goes to $A_1^*\shuffle\ldots\shuffle A_k^*$ ($k$ clauses of type 3).
\rnew{Essentially, we distribute the $*$ of a disjunction of atoms with intervals to each of the atoms.
For example, $(a\mid (b\shuffle c^?))^*$ goes to $a^*\shuffle (b\shuffle c^?)^*$.
}

\item $(A_1^{I_1}\mid\ldots\mid A_k^{I_k})^+$ goes to $A_1^*\shuffle\ldots\shuffle A_k^*$ ($k$ clauses of type 3) if \rnew{there exists an atom with interval $A_i^{I_i}$ ($i\in\{1,\ldots,k\}$) that defines the empty word i.e., $I_i= ?$ or \rnew{$I^a =?$} for every symbol $a\in\Sigma_{A_i}$.
If the empty word is defined, then we can basically transform the $+$ into $*$ and then distribute the $*$ as for the previous case.
For example, $((a\shuffle b^?)\mid (c\shuffle d)^?)^+$ goes to $(a\shuffle b^?)^*\shuffle (c\shuffle d)^*$.
}

\end{itemize}
\item Take a clause $(A_1^{I_1}\mid\ldots\mid A_k^{I_k})$.
\begin{itemize}
\item $(A_1^{I_1}\mid\ldots\mid A_k^{I_k})^?$ goes to $(A_1^{I_1^?}\mid\ldots\mid A_k^{I_k^?})$ (type 3).
\rnew{We essentially distribute the $?$ of a disjunction of atoms with intervals to each of the atoms.
For example, $(a^{[2,3]}\mid b^+)^?$ goes to $(a^{[2,3]^?}\mid b^*)$.
}

\item $(A_1^{I_1}\mid\ldots\mid A_k^{I_k})$ goes to $(A_1^{I_1^?}\mid\ldots\mid A_k^{I_k^?})$ (type 3) if \rnew{there exists an atom with interval $A_i^{I_i}$ ($i\in\{1,\ldots,k\}$) that defines the empty word i.e., $0$ belongs to the set represented by $I_i$ or \rnew{$I^a =?$} for every symbol $a\in\Sigma_{A_i}$. 
If the empty word is defined by one of the atoms, then we can basically distribute $?$ to all of them.
For example, $(a\mid (b\shuffle c)^{[0,5]})$ goes to $(a^?\mid (b\shuffle c)^{[0,5]})$.
}
\end{itemize}

\item \rnew{Take an atom $(a_1^?\shuffle \ldots \shuffle a_k^?)$ and an interval $I$.
Then, $(a_1^?\shuffle \ldots \shuffle a_k^?)^I$ goes to $(a_1^?\shuffle\ldots\shuffle a_k^?)^{[0,\max(I)]}$, where by $\max(I)$ we denote the maximum value from the set represented by the interval $I$.
This step may be combined with one of the previous ones to rewrite a clause with interval as one of type 3.
For example, $((a^?\shuffle b^?)^{[3,6]}\mid c)$ goes to $((a^?\shuffle b^?)^{[0,6]}\mid c^?)$.
}

\item Remove symbols $a$ (resp.\ atoms $A$ or clauses $D$) such that $I^a$ (resp.\ $I^A$ or $I^D$) is $[0,0]$.
\end{itemize}
Note that each of the rewriting steps gives an equivalent reduced expression. 

Next, we assume that we work with reduced DIMEs only and show that the language defined by a DIME $E$ comprises of all unordered words satisfying the characterizing tuple $\Delta_E$.

\begin{lemma}\label{lemmawordtuple}
Given an unordered word $w$ and a DIME $E$, $w\in L(E)$ iff $w\models\Delta_E$.
\end{lemma}
\begin{proof}
The \emph{only if} part follows from the definition of the satisfiability of $\Delta_E$.
For the \emph{if} part, we take the tuple $\Delta_E$ corresponding to a DIME $E=D_1^{I_1}\shuffle\ldots\shuffle D_k^{I_k}$ and an unordered word $w$ such that $w\models \Delta_E$.
Let $w=w_1\uplus\ldots\uplus w_k\uplus w'$, where each $w_i$ contains all occurrences in $w$ of the symbols from $\Sigma_{D_i}$ (for $1\leq i\leq k$).
Since $w\models N_E$, we infer that there is no symbol $a\in\Sigma\setminus(\Sigma_{D_1}\cup\ldots\cup\Sigma_{D_k})$ such that $a\in w$, which implies $w'=\varepsilon$.
Thus, proving $w\models E$ reduces to proving that $w_i\models D_i^{I_i}$ (for $1\leq i\leq k$).
Since $E$ is a reduced DIME, each derivation can be constructed by reasoning on the three possible types of the $D_i^{I_i}$ (for $1\leq i\leq k$).

\rnew{
{\em Case 1.} 
Take $D_i^{I_i}=(A_1\mid\ldots\mid A_k)^+$ of type 1.
From the semantics of the UREs, we observe that proving $w_i\models D_i^{I_i}$ is equivalent to proving that (i) $w_i$ is non-empty and (ii) $w_i$ can be split as $w_i= w_1'\uplus\ldots\uplus w_p'$, where every $w_j'$ ($1\leq j\leq p$) satisfies an atom $A_l$ ($1\leq l\leq k$).
First, we point out that since $w$ satisfies the collections of required symbols $P_E$, we infer that $w_i$ is non-empty, which implies (i).
Then, since $w$ satisfies the extended cardinality map $N_E$ and the counting dependencies $K_E$, we infer that (ii) is also satisfied.

{\em Case 2.} 
Take $D_i^{I_i} = (A_1^{I_1}\mid\ldots\mid A_k^{I_k})$ of type 2.
From the semantics of UREs, we observe that proving $w_i\models D_i^{I_i}$ is equivalent to proving that (i) $w_i$ is non-empty and (ii) there exists an atom with interval $A_j^{I_j}$ ($1\leq j\leq k$) such that $w_i\models A_j^{I_j}$.
Since $w\models P_E$, we infer that $w_i$ is non-empty hence (i) is satisfied.
Then, since $w\models C_E$, we infer that only the symbols from one atom $A_j$ of $D_i$ are present in $w_i$.
Moreover, since  $w\models N_E$ and $w\models K_E$, we infer that the number of occurrences of each symbol from $\Sigma_{A_j}$ are such that $w_i\models A_j^{I_j}$.
Hence, the condition (ii) is also satisfied.

{\em Case 3.} 
Take $D_i^{I_i} = (A_1^{I_1}\mid\ldots\mid A_k^{I_k})$ of type 3.
The only difference w.r.t.\ the previous case is that  $w_i$ may be also empty, hence proving $w_i\models D_i^{I_i}$ is equivalent to proving only that there exists an atom with interval $A_j^{I_j}$ ($1\leq j\leq k$) such that $w_i\models A_j^{I_j}$, which follows similarly to the previous case.
}
\qed\end{proof}
Moreover, we define the \emph{subsumption} of two characterizing tuples, which captures the containment of DIMEs.
Given two DIMEs $E$ and $E'$, we write $\Delta_{E'}\preccurlyeq\Delta_{E}$ if $C_{E}\subseteq C_{E'}$, $N_{E'}\subseteq N_{E}$, $P_{E}\subseteq P_{E'}$, and $K_{E}\subseteq K_{E'}$. 
Then, we obtain the following.

\begin{lemma}\label{lemmacnttuples}
Given two DIMEs $E$ and $E'$, ${L}(E')\subseteq {L}(E)$ iff $\Delta_{E'}\preccurlyeq\Delta_{E}$.
\end{lemma}

\begin{proof}
First, we claim that given two DIMEs $E$ and $E'$:
$\Delta_{E'}\preccurlyeq\Delta_E$ iff $w\models \Delta_{E'}$ implies $w\models\Delta_E$ for every $w$ (*).
The \emph{only if} part of (*) follows directly from the definitions while the \emph{if} part can be easily shown by contraposition.
From Lemma~\ref{lemmawordtuple} and (*) we infer the correctness of Lemma~\ref{lemmacnttuples}.
\qed\end{proof}
\begin{example}\normalfont
For the following DIMEs, it holds that $L(E')\subsetneq L(E)$ and $L(E)\varnot\subseteq L(E')$:
\begin{itemize}
\item Take $E=a^*\shuffle b^*$ and $E' = (a\shuffle b^?)^*$. 
Note that $K_E=\emptyset$ and $K_{E'} = \{(a,b)\}$. 
For instance, the unordered word $b$ belongs to $L(E)$, but does not belong to $L(E')$\rnew{.}
\item Take $E=a^{[3,6]^?}\mid b^*$ and $E' = a^{[3,6]}\mid b^+$. 
Note that $P_{E} = \emptyset$, and $P_{E'} = \{\{a,b\}\}$. 
For instance, the unordered word $\varepsilon$ belongs to $L(E)$, but does not belong to $L(E')$\rnew{.}
\item Take $E=(a\shuffle b^?)^*$ and $E' = (a\shuffle b^?)^{[0,5]}$.
Note that $(a,6)$ belongs to $N_{E}$, but not to $N_{E'}$.
For instance, the unordered word $a^6$ belongs to $L(E)$, but does not belong to $L(E')$\rnew{.}
\item Take $E=(a\mid b)^+$ and $E' = a^+\mid b^+$.  Note that $C_{E} = \emptyset$, and $C_{E'} = \{(a,b),(b,a)\}$. 
For instance, the unordered word $ab$ belongs to $L(E)$, but does not belong to $L(E')$.\qed
\end{itemize}
\end{example}
Lemma~\ref{lemmacnttuples} shows that two equivalent DIMEs yield the same characterizing tuple, and hence, the tuple $\Delta_E$ can be viewed as a ``canonical form'' for the language defined by a DIME $E$.
\rnew{Formally, we obtain the following.}
\begin{corollary}
Given two DIMEs $E$ and $E'$, $L(E) = L(E')$ iff $\Delta_{E} = \Delta_{E'}$.
\end{corollary}
In the next section, we show that the characterizing tuple has a compact representation that permits us to decide the problems of membership and containment in polynomial time.

\subsection{Tractability of DIMEs}\label{subsec:tractability-dime}

We now show that the characterizing tuple admits a compact representation that yields the tractability of deciding membership and containment of DIMEs.

Given a reduced DIME $E$, note that $C_E$ and $K_E$ are quadratic in $|\Sigma|$ and can be easily constructed. 
\rnew{The set} $C_E$ consists of all pairs of distinct symbols $(a,b)$ such that they appear in different atoms in the same clause of type 2 or 3.
Moreover, $K_E$ consists of all pairs of distinct symbols $(a,b)$ such that they appear in the same atom and $I^a = 1$.

While $N_E$ may be infinite, it can be easily represented in a compact manner using intervals: for every symbol $a$, the set $\{i\in \mathbb N_0\mid(a,i)\in N_E\}$ is representable by an interval. 
Given a symbol $a\in\Sigma$, by $\hat N_E(a)$ we denote the interval representing the set $\{i\in\mathbb N_0\mid (a,i)\in N_E\}$ that can be easily obtained from $E$:

\begin{itemize}
\item $\hat N_E(a) = [0,0]$ if $a$ appears in no clause in $E$,
\item $\hat N_E(a) = [0,\infty]$ (or simply $*$) if $a$ appears in a clause of type 1 in $E$,
\item $\hat N_E(a)= I^A$ if $I^a = 1$, $A$ is the atom containing $a$, and $A$ is the unique atom of a clause of type 2 or 3,
\item $\hat N_E(a) = {I^A}^?$ if $I^a = 1$, $A$ is the atom containing $a$, and $A$ appears in a clause of type 2 or 3 containing at least two atoms,
\item $\hat N_E(a)= [0,\max(I^A)]$ if $I^a = ?$, $A$ is the atom containing $a$, and $A$ appears in a clause of type 2 or 3.
\end{itemize}  
For example, for $E_0= a^+
\shuffle ((b \shuffle c^?)^+\mid d^{[5,\infty]})$, we obtain the following $\hat N_{E_0}$:
\begin{gather*}
 \hat N_{E_0}(a)=+, \qquad \hat N_{E_0}(b)\,=*, \qquad
 \hat N_{E_0}(c)\,=*, \qquad \hat N_{E_0}(d)=[5,\infty]^?.
\end{gather*}
Naturally, testing $N_{E'}\subseteq N_{E}$ reduces to a simple test on $\hat N_{E'}$ and $\hat N_{E}$.

\rnew{
Representing $P_E$ in a compact manner is more tricky. 
A natural idea would be to store only its $\subseteq$-minimal elements since $P_E$ is closed under supersets. 
Unfortunately, there exist DIMEs having an exponential number of $\subseteq$-minimal elements. 
For instance, for the DIME $E_1=((a\shuffle b) \mid (c\shuffle d))^+\shuffle ((e\shuffle f)^{[2,5]}\mid g^{[1,3]})\shuffle (h^*\shuffle i^{[0,9]})$, the set $P_{E_1}$ has 6 $\subseteq$-minimal elements: $\{a,c\}$, $\{a,d\}$, $\{b,c\}$, $\{b,d\}$, $\{e,g\}$, and $\{f,g\}$.
The example easily generalizes to arbitrary numbers of atoms used in the clauses.

However, we observe that the exponentially-many $\subseteq$-minimal elements may contain redundant information that is already captured by other elements of the characterizing tuple.
For instance, for the above DIME $E_1$, if we know that $\{a,c\}$ belongs to $P_E$, we can easily see that other $\subseteq$-minimal elements also belong to $P_E$.
More precisely, we observe that for every unordered word $w$ defined by $E$ it holds that $w(a) = w(b)$, $w(c)=w(d)$ and $w(e)=w(f)$, which is captured by the counting dependencies $K_E=\{(a,b),(b,a),(c,d),(d,c),(e,f),(f,e)\}$.
Hence, for the unordered words defined by $E$, the presence of an $a$ implies the presence of a $b$, the presence of a $c$ implies the presence of a $d$, etc.
Consequently, if $\{a,c\}$ belongs to $P_E$, then $\{b,c\}$, $\{a,d\}$, and $\{b,d\}$ also belong to $P_E$.
Similarly, if $\{e,g\}$ belongs to $P_E$, then $\{f,g\}$ also belongs to $P_E$.

Next, we use the aforementioned observation to define a compact representation of $P_E$.
For this purpose, we introduce the auxiliary notion of symbols \emph{implied by a DIME $E$ in the presence of a set of symbols $X$}, denoted $\impl_E(X)$:
\[
\impl_E(X)=X\cup\{a\in\Sigma\mid \exists b\in X.\ (a,b)\in K_E \text{ and } (b,a)\in K_E\}.
\]
For example, for the above $E_1$, we have $\impl_E(\{a,c\}) = \{a,b,c,d\}$.

Moreover, given a DIME $E$, by $\pmin$ we denote the set of all $\subseteq$-minimal elements of $P_E$.
Given a subset $P\subseteq\pmin$, we say that $P$ is:
\begin{itemize}
\item \emph{non-redundant} if $\forall X\in P.\ \not\exists Y\in P.\ X\subseteq\impl_E(Y)$,
\item \emph{covering} if $\forall X\in\pmin.\ \exists Y\in P.\ X\subseteq \impl_E(Y)$.
\end{itemize}
For example, take the above $E_1=((a\shuffle b) \mid (c\shuffle d))^+\shuffle ((e\shuffle f)^{[2,5]}\mid g^{[1,3]})\shuffle (h^*\shuffle i^{[0,9]})$ and recall that $\pminpp=\{$$\{a,c\}$,$\{a,d\}$,$\{b,c\}$,$\{b,d\}$,$\{e,g\}$,$\{f,g\}$$\}$. 
Then, we have the following:

\begin{itemize}
\item $\{\{b,c\},\{f,g\}\}$ is non-redundant and covering,
\item $\{\{b,c\}\}$ is non-redundant and it is not covering,
\item $\{\{a,c\},\{b,c\},\{f,g\}\}$ is redundant and covering,
\item $\{\{a,c\},\{b,c\}\}$ is redundant and not covering.
\end{itemize}
Given a DIME $E$, the \emph{compact representation of the collections of required symbols} $P_E$ is naturally a non-redundant and covering subset of $\pmin$.
Since there may exist many non-redundant and covering subsets of $\pmin$, we use the total order $<_\Sigma$ on the alphabet $\Sigma$ to propose a deterministic construction of the compact representation $\hat P_E$.
For this purpose, we define first some additional notations.

Given an atom $A$, by $\Phi(A)$ we denote the smallest label from $\Sigma$ w.r.t.\ $<_\Sigma$ that is present in $A$ and has interval $1$:
\[
 \Phi(A)=\min_{<_\Sigma}\{a\in\Sigma_A\mid I^a = 1\}.
\]
For example, $\Phi(a\shuffle b) = a$.
Then, given a clause with interval $D^I$, by $\Phi(D^I)$ we denote the set of all symbols $\Phi(A)$ for every atom $A$ in $D$:
\[
 \Phi(D^I)=\{\Phi(A)\mid A\text{ is an atom in }D\}.
\]
For example, $\Phi(((a\shuffle b) \mid (c\shuffle d))^+) = \{a,c\}$ and $\Phi(((e\shuffle f)^{[2,5]}\mid g^{[1,3]}))=\{e,g\}$.
Then, $\hat P(E)$ consists of all such sets for the clauses with intervals of type 1 or 2:
\[
 \hat P_E = \{\Phi(D^I)\mid D^I \text{ is a clause with interval of type 1 or 2 in }E\}.
\]
For example, $\hat P_{E_1}=\{\{a,c\},\{e,g\}\}$.
Notice that the set $\{a,c\}$ is due to the clause with interval $((a\shuffle b) \mid (c\shuffle d))^+$ of type 1 and the set $\{e,g\}$ is due to the clause with interval $((e\shuffle f)^{[2,5]}\mid g^{[1,3]})$ of type 2.
Also notice that the clause with interval $(h^*\shuffle i^{[0,9]})$ is of type 3, none of its symbols is required, and consequently, no set in $\hat P_E$ contains symbols from it.

We have introduced all elements to be able to define the compact representation of a characterizing tuple.
Given a DIME $E$, we say that $\hat\Delta=(C_E,\hat N_E,\hat P_E, K_E)$ is the \emph{compact representation} of its characterizing tuple $\Delta_E$.
Then, an unordered word $w$ satisfies $\hat\Delta_E$, denoted $w\models\hat \Delta_E$, if 
\begin{itemize}
\item $w\models C_E$ and $w\models K_E$ as previously defined when we have introduced $w\models\Delta_E$,
\item $w\models \hat N_E$ i.e., $\forall a\in\Sigma.\ w(a)\in\hat N_E(a)$,
\item $w\models \hat P_E$ i.e., $\forall X\in \hat P_E.\ \exists a\in X.\ a\in w$.
Notice that we use exactly the same definition as for $w\models P_E$ and recall that $\hat P_E$ is in fact a non-redundant and covering subset of $\pmin$.
\end{itemize}
Next, we show that given a DIME $E$, its compact characterizing tuple $\hat\Delta_E$ defines precisely the same set of unordered words as its characterizing tuple $\Delta_E$.

\begin{lemma}\label{lemma:satisfy:compact}
Given an unordered word $w$ and a DIME $E$, $w\models \Delta_E$ iff $w\models \hat\Delta_E$.
\end{lemma}
\begin{proof}
The \emph{only if} part follows directly from the definitions.
For the \emph{if} part, proving $w\models \Delta_E$ reduces to proving that $w\models P_E$, which moreover, reduces to proving that for every $X$ from $\pmin$ there is a symbol $a$ in $X$ that occurs in $w$ (*).
Since $\hat P_E$ is a covering subset of $\pmin$, we know that for every $X\in \pmin$ there exists a set $Y\in\hat P_E$ such that $X\subseteq\impl_E(Y)$.
Since $w\models\hat P_E$ and $w\models K_E$, we infer that (*) is satisfied.
\qed\end{proof}
Additionally, we define the \emph{subsumption} of the compact representations of two characterizing tuples.
Given two DIMEs $E$ and $E'$, we write $\hat\Delta_{E'}\preccurlyeq\hat\Delta_E$ if 

\begin{itemize}
\item $C_E\subseteq C_{E'}$ and $K_E\subseteq K_{E'}$ (as for the subsumption of characterizing tuples),
\item $\forall a\in \Sigma.\ \hat N_{E'}(a)\subseteq \hat N_E(a)$,
\item $\forall X\in\hat P_E.\ \exists Y\in \hat P_{E'}.\ Y\subseteq \impl_{E'}(X)$.
\end{itemize}
Next, we show that the subsumption of compact representations of characterizing tuples captures the subsumption of characterizing tuples.

\begin{lemma}\label{lemma:equivalent:compact}
Given two DIMEs $E$ and $E'$, $\Delta_{E'}\preccurlyeq \Delta_E$ iff $\hat\Delta_{E'}\preccurlyeq \hat\Delta_E$.
\end{lemma}
\begin{proof}
First, since $P_E$ is closed under supersets, we observe that 
\[
 P_E\subseteq P_{E'} \text{ iff } \forall X\in\pmin.\ \exists Y\in\pminp.\ Y\subseteq X.
\]  
Moreover, the conditions $C_E\subseteq C_{E'}$ and $K_E\subseteq K_{E'}$ are part of both $\Delta_{E'}\preccurlyeq \Delta_E$ and $\hat\Delta_{E'}\preccurlyeq \hat\Delta_E$.
Consequently, proving $\Delta_{E'}\preccurlyeq \Delta_E$ iff $\hat\Delta_{E'}\preccurlyeq \hat\Delta_E$ reduces to proving that, if $C_E\subseteq C_{E'}$ and $K_E\subseteq K_{E'}$, then
\[
\forall X\in\pmin.\ \exists Y\in\pminp.\ Y\subseteq X \text{ iff } \forall X\in \hat P_E.\ \exists Y \in\hat P_{E'}.\ Y\subseteq \impl_{E'}(X).
\]
For the \emph{only if} part, take a set $X$ from $\hat P_E$.
Since $X$ also belongs to $\pmin$, we know by hypothesis that there exists a set $Y$ in $\pminp$ such that $Y\subseteq X$.
Then, construct a set $Y'$ from $Y$ by replacing each symbol $b$ from $Y$ with the smallest $a$ w.r.t.\ $<_\Sigma$ such that $(a,b)$ and $(b,a)$ belong to $K_{E'}$.
Moreover, since $K_E\subseteq K_{E'}$, we infer that $Y'\subseteq \impl_{E'}(X)$.
For the \emph{if} part, take an $X$ from $\hat P_E$ and an $Y$ from $\hat P_{E'}$ s.t.\ $Y\subseteq \impl_{E'}(X)$.
To construct the corresponding $X'$ in $\pmin$ and $Y'$ in $\pminp$ such that $Y'\subseteq X'$, we replace symbols $a$ from $X$ and $a'$ from $Y$ with symbols $b$ in $X'$ and $b'$ in $Y'$ such that $(a,b)$ and $(b,a)$ belong to $K_E$, and $(a',b')$ and $(b',a')$ belong to $K_{E'}$.
Since $K_E\subseteq K_{E'}$, we know that such $X'$ and $Y'$ do exist.
\qed\end{proof}

\begin{example}

Take $E = a^*\shuffle (b\mid c)^+\shuffle d^*$ and $E'=(a\shuffle b)^+\mid (c\shuffle d)^+$.
Notice that $L(E')\subseteq L(E)$, $\Delta_{E'}\preccurlyeq \Delta_E$, and $\hat\Delta_{E'}\preccurlyeq\hat\Delta_E$.
In particular, we have the following. 
\begin{itemize}
\item $C_E=\emptyset$ is included in $C_{E'}=\{(a,c),(a,d),(b,c),(b,d),(c,a),(c,b),(d,a),(d,b)\}$,
\item $\hat N_E(a) = \hat N_{E'}(a) = *,\ldots, \hat N_E(d) = \hat N_{E'}(d) = *$, 
\item $K_E=\emptyset$ is included in $K_{E'} = \{(a,b),(b,a),(c,d),(d,c)\}$,
\item $\hat P_E = \{\{b,c\}\}$ and $\hat P_{E'} = \{\{a,c\}\}$ that compactly represent $P_E=\{\{b,c\},\ldots\}$ and $P_{E'}=\{\{a,c\}, \{a,d\}, \{b,c\}, \{b,d\},\ldots\}$, respectively (we have listed only the $\subseteq$-minimal sets).
Then, take $X=\{b,c\}$ from $\hat P_E$ and notice that there exists $Y=\{a,c\}$ in $\hat P_{E'}$ such that $
Y\subseteq \impl_{E'}(X)$ because $\impl_{E'}(\{b,c\}) = \{a,b,c,d\}$.\qed
\end{itemize}
\end{example}
Next, we show that the compact representation is of polynomial size.

\begin{lemma}\label{lemma:polynomial:compact}
Given a DIME $E$, the compact representation $\hat\Delta_E=(C_E,\hat N_E,\hat P_E,K_E)$ of its characterizing tuple $\Delta_E$ is of size polynomial in the size of the alphabet $\Sigma$.
\end{lemma}
\begin{proof}
By construction, the sizes of $C_E$ and $K_E$ are quadratic in $|\Sigma|$ while the sizes of $\hat P_E$ and $\hat N_E$ are linear in $|\Sigma|$.
\qed\end{proof}
The use of compact representation of characterizing tuples allows us to state the main result of this section.

\begin{theorem}\label{thdimeptime}
Given an unordered word $w$ and two DIMEs $E$ and $E'$:
\begin{enumerate}
\item deciding whether $w\in L(E)$ is in PTIME,
\item deciding whether $L(E')\subseteq L(E)$ is in PTIME.
\end{enumerate}
\end{theorem}

\begin{proof} 
The first part follows from Lemma~\ref{lemmawordtuple}, Lemma~\ref{lemma:satisfy:compact}, and Lemma~\ref{lemma:polynomial:compact}.
The second part follows from Lemma~\ref{lemmacnttuples}, Lemma~\ref{lemma:equivalent:compact}, and Lemma~\ref{lemma:polynomial:compact}.
\qed\end{proof}

}

\section{Interval multiplicity schemas}\label{sec:static}
In this section, we employ DIMEs to define schema languages and we present the related problems of interest.
\begin{definition}
A \emph{disjunctive interval multiplicity schema (DIMS)} is a tuple $S=(\root_S, R_S)$, where $\root_S\in\Sigma$ is a designated root label and $R_S$ maps symbols in $\Sigma$ to DIMEs.  
By $\dims$ we denote the set of all disjunctive interval multiplicity schemas.  
A \emph{disjunction-free interval multiplicity schema (IMS)} $S=(\root_S, R_S)$ is a restricted DIMS, where $R_S$ maps symbols in $\Sigma$ to IMEs.
By $\sims$ we denote the set of all disjunction-free interval multiplicity schemas.
\end{definition}
We define the language captured by a DIMS $S$ in the following way.
Given a tree $t$, we first define the unordered word $ch_t^n$ of children of a node $n\in N_t$ of
$t$ i.e., $\ch_t^n(a) = |\{m\in N_t\mid(n,m)\in\child_t\wedge
\lab_t(m)=a\}|$. Now, a tree $t$ \emph{satisfies} $S$, in symbols $t \models S$, if $\lab_t(\root_t) = \root_S$ and for every node $n\in N_t$, $\ch_t^n\in L(R_S(\lab_t(n)))$.
By $L(S)\subseteq \Tree$ we denote the set of all trees satisfying $S$.

In the sequel, we present a schema $S=(\root_S,R_S)$ as a set of rules of the form $a\rightarrow R_S(a)$, \rnew{for every} $a\in\Sigma$.
If $L(R_S(a)) = \varepsilon$, then we write $a\rightarrow \epsilon$ or we simply omit writing such a rule.

\begin{example}
\label{ex:1}
\normalfont
Take the content model of a semi-structured database storing information about a peer-to-peer file sharing system, having the following rules:
\emph{1)} a peer is allowed to download at most the same number of files that \rnew{it} uploads, and \emph{2)} peers are split into two groups: a peer is a \emph{vip} if \rnew{it} uploads at least 100 files, otherwise it is a simple \emph{user}:
\begin{flalign*}
\mathit{peers} \rightarrow~ & \mathit{user}^*\shuffle\mathit{vip}^*,\\
\mathit{user} \rightarrow~ & (\mathit{upload}\shuffle\mathit{download}^?)^{[0,99]},\\
\mathit{vip} \rightarrow~ &  (\mathit{upload}\shuffle\mathit{download}^?)^{[100,\infty]}.\tag*{\qed}
\end{flalign*}
\end{example}
\begin{example}\normalfont
\label{ex:2}
Take the content model of a semi-structured database storing information about two types of cultural events: plays and movies.
Every event has a date when it takes place.
If the event is a play, then it takes place in a theater while a movie takes place in a cinema.
\begin{flalign*}
\mathit{events}\rightarrow~ & \mathit{event^*},\\
\mathit{event}\rightarrow~ & \mathit{date\shuffle ((play\shuffle theater)\mid (movie\shuffle cinema))}.\tag*{\qed}
\end{flalign*}
\end{example}

\paragraph{\rnew{{\bf Problems of interest.}}}
We define next the problems of interest and we formally state the corresponding decision problems parameterized by the class of schema $\mathcal S$ and, when appropriate, by a class of queries $\mathcal Q$.
\begin{itemize}
\item {\em Schema satisfiability} -- checking if there exists a tree satisfying the given schema:
\[
 \SAT_{\mathcal{S}} = \{ S \in\mathcal{S} \mid \exists t\in \Tree.\ t\models S \}.
\]
\item {\em Membership} -- checking if the given tree satisfies the given schema:
\[
 \MEMB_{\mathcal{S}} = \{(S,t) \in\mathcal{S}\times\Tree\mid t\models S\}.
\]
\item {\em Schema containment} -- checking if every tree satisfying one given schema satisfies another given schema:
\[
 \CNT_{\mathcal{S}} = \{(S_1,S_2)\in\mathcal{S}\times\mathcal{S}\mid L(S_1)\subseteq L(S_2)\}.
\]
\item {\em Query satisfiability by schema} -- checking if there exists a tree that satisfies the given schema and the given query:
\[
 \SAT_{\mathcal{S},\mathcal{Q}} = \{(S,q)\in \mathcal{S}\times\mathcal{Q}\mid \exists t\in L(S).\ t\models q\}.
\]
\item {\em Query implication by schema} -- checking if every tree satisfying the given schema satisfies also the given query:
\[
 \IMPL_{\mathcal{S},\mathcal{Q}} = \{(S, q)\in \mathcal{S}\times\mathcal{Q}\mid \forall t\in L(S).\ t\models q\}.
\]
\item {\em Query containment in the presence of schema} -- checking if every tree satisfying the given schema and one given query also satisfies another given query: 
\[
 \CNT_{\mathcal{S},\mathcal{Q}}= \{(p, q, S) \in \mathcal{Q}\times\mathcal{Q}\times\mathcal{S} \mid \forall t\in L(S).\ t \models p\Rightarrow t\models q\}.
\]
\end{itemize}
\rnew{We study these problems for DIMSs and IMSs in Sections~\ref{dims} and~\ref{ms} of the paper.}

\section{Complexity of disjunctive interval multiplicity schemas (DIMSs)}\label{dims}
In this section, we present the complexity results for DIMSs.
First, we show the tractability of schema satisfiability and containment.
Then, we provide an algorithm for deciding membership in \emph{streaming} i.e., that processes an XML document in a single pass and using memory depending on the height of the tree and not on its size.
Finally, we point out that the complexity of query satisfiability, implication, and containment in the presence of the schema follow from existing results.

First, we show the tractability of schema satisfiability and schema containment.

\rnew{\begin{proposition}\label{dims:ptime}
$\SAT_\dims$ and $\CNT_{\dims}$ are in PTIME.
\end{proposition}}

\begin{proof}
A simple algorithm based on dynamic programming can decide the satisfiability of a DIMS.
\rnew{More precisely, given a schema $S=(\root_S,R_S)$, one has to determine for every symbol $a$ of the alphabet $\Sigma$ whether there exists a (finite) tree $t$ that satisfies $S'=(a,R_S)$.
Then, the schema $S$ is satisfiable if there exist such a tree for the root label $\root_S$.
}

Moreover, testing the containment of two DIMSs reduces to testing, for each symbol in the alphabet, the containment of the associated DIMEs, which is in PTIME (Theorem~\ref{thdimeptime}).
\qed\end{proof}
\noindent Next, we provide an algorithm for deciding membership in \emph{streaming} i.e., that processes an XML document in a single pass and uses memory depending on the height of the tree and not on its size.
\rnew{Our notion of streaming has been employed in~\cite{SeVi02} as a relaxation of the constant-memory XML validation against DTDs, which can be performed only for some DTDs~\cite{SeVi02,SeSi07}.
In general, validation against DIMSs cannot be performed with constant memory due to the same observations as in~\cite{SeVi02,SeSi07} w.r.t.\ the use of recursion in the schema. 
Hence, we have chosen our notion of streaming to be able to have an algorithm that works for the entire class of DIMSs.
}
We assume that the input tree is given in XML format, with arbitrary ordering of sibling nodes.
\rnew{
Moreover, the proposed algorithm has \emph{earliest rejection} i.e., if the given tree does not satisfy the given schema, the algorithm outputs the result as early as possible.
}
For a tree $t$, $\mathit{height}(t)$ is the height of $t$ defined in the usual way. 
We employ the standard RAM model and assume that subsequent natural numbers are used as labels in $\Sigma$. 

\rnew{\begin{proposition}\label{dims:ptime2}
$\MEMB_\dims$ is in PTIME.
There exists an earliest rejection streaming algorithm that checks membership of a tree $t$ in a DIMS $S$ in time $O(|t|\times|\Sigma|^2)$ and using space $O(\mathit{height}(t) \times|\Sigma|^2)$.
\end{proposition}}
\begin{proof}
We propose Algorithm~\ref{alg1} for deciding the membership of a tree $t$ to the language of a DIMS $S$.
The input tree $t$ is given in XML format, with some arbitrary ordering of sibling nodes.
We assume a well-formed stream $\widetilde{t} \subset\{\open,\close\}\times\Sigma$ representing a tree $t$ and a procedure $\readb(\widetilde{t})$ that returns the next pair $(\theta,b)$ in the stream, where $\theta\in\{\open,\close\}$ and $b\in\Sigma$.
The algorithm works for every arbitrary ordering of sibling nodes.
To validate a tree $t$ against a DIMS $S = (\root_S, R_S)$, one has to run Algorithm~\ref{alg1} after reading the opening tag of the root.

\rnew{
For a given node, the algorithm constructs the compact representation of the characterizing tuple of its label (line 1), which requires space $O(|\Sigma|^2)$ (cf.\ Lemma~\ref{lemma:polynomial:compact}).
The algorithm also stores for a given node the number of occurrences of each label in $\Sigma$ among its children.
This is done using the array $\counta$, which requires space $O(\Sigma)$.
Initially, all values in the array $\counta$ are set at 0 (lines 2-3) and they are updated after reading the open tag of the children (lines 4-6).
During the execution, the algorithm maintains a stack whose height is the depth of the currently visited node.
Naturally, the bound on space required is $O(\mathit{height}(t)\times|\Sigma|^2)$.

The algorithm has earliest rejection since it rejects a tree as early as possible.
More precisely, this can be done after reading the opening tag for nodes that violate the maximum value for the allowed cardinality for their label (lines 7-8) or violate some conflicting pair of siblings (lines 9-10).
If it is not the case, the algorithm recursively validates the corresponding subtree (lines 11-12).
After reading all children of the current node, the algorithm checks whether the components of the characterizing tuple are satisfied: the extended cardinality map (lines 14-15), the collections of required symbols (lines 16-17), and the counting dependencies (lines 18-19).
Notice that since we have checked the conflicting pairs of siblings after reading each opening tag, we do not need to check them again after reading all children.
However, we still need to check the extended cardinality map at this moment to see whether the number of occurrences of each label is in the allowed interval.
When we have read the opening tag, we were able to reject only if the maximum value for the allowed number of occurrences has been already violated.
As for the collections of required symbols and the counting dependencies, we are able to establish whether they are satisfied or not after reading all children.
If none of the constraints imposed by the characterizing tuple is violated, the algorithm returns true (line 20).
As we have already shown with Lemma~\ref{lemmawordtuple} and Lemma~\ref{lemma:satisfy:compact}, the compact representation of the characterizing tuple captures precisely the language of a given DIME.
Consequently, the algorithm returns true after reading the root node iff the given tree satisfies the given schema.
}
\qed \end{proof}

\begin{algorithm}\caption{Streaming algorithm for testing membership.\label{alg1}}
\ALGORITHM \emph{validate}$(a)$\\
{\bf Parameters:} DIMS $S$, stream $\widetilde{t}$\\
\INPUT: the label $a\in \Sigma$ of the current node\\
\OUTPUT: $\true$ if the subtree rooted at the current node is valid w.r.t.\ $S$, $\false$ otherwise\\
\LN \LET $(C,\hat N,\hat P,K)$ be the compact representation of the characterizing tuple of $R_S(a)$\\
\LN \FOR $b\in\Sigma$ \DO\\
\LN \TAB \LET $\counta[b] = 0$\\
\LN $(\theta,b) = \readb(\widetilde{t})$\\
\LN \WHILE $\theta=\open$ \DO\\
\LN \TAB $\counta[b]\colonequals \counta[b] + 1$\\
\LN \TAB \IF $\counta[b] > \max(\hat N(b))$ \THEN\\
\LN \TAB \TAB \RETURN $\false$\\
\LN \TAB \IF $\exists c\in\Sigma.\ (b,c)\in C\wedge \counta[c]\neq 0$ \THEN \\
\LN \TAB \TAB \RETURN $\false$\\
\LN \TAB \IF \emph{validate}$(b) = \false$ \THEN\\
\LN \TAB \TAB \RETURN $\false$\\
\LN \TAB $(\theta,b) = \readb(\widetilde{t})$\\
\LN \IF $\exists b\in\Sigma.\ \counta[b]\notin \hat N(b)$ \THEN\\
\LN \TAB \RETURN $\false$\\
\LN \IF $\exists X\in \hat P.\ \forall b\in X.\ \counta[b] = 0$ \THEN \\
\LN \TAB \RETURN $\false$\\
\LN \IF $\exists (b,c)\in K.\ \counta[b]<\counta[c]$ \THEN\\
\LN \TAB \RETURN $\false$\\
\LN \RETURN $\true$
\end{algorithm}
\noindent We continue with complexity results that follow from known facts.
Query satisfiability for DTDs is NP-complete~\cite{BeFaGe08} and we adapt the result for DIMSs.

\begin{proposition}\label{propsatqnphard}
$\SAT_{\dims,\Twig}$ is NP-complete.
\end{proposition}
\begin{proof}
Proposition 4.2.1 from~\cite{BeFaGe08} implies that satisfiability of twig queries in the presence of DTDs is NP-hard.
We adapt the proof and we obtain the following reduction from SAT to $\SAT_{\dims, \Twig}$: we take a CNF formula $\varphi=\bigwedge_{i=1}^nC_i$ over the variables $x_1,\ldots,x_m$, where each $C_i$ is a disjunction of literals. 
We take $\Sigma=\{r, t_1,f_1,\ldots,t_m,f_m,c_1,\ldots,c_n\}$ and we construct:
\begin{itemize}
\item The DIMS $S$ having the root label $r$ and the rules:
\begin{itemize}
\item $r \rightarrow (t_1 \mid f_1) \shuffle \dots\shuffle(t_m\mid f_m)$,
\item $t_j \rightarrow c_{j_1}\shuffle\ldots\shuffle c_{j_k}$, where $c_{j_1},\ldots,c_{j_k}$ correspond to the clauses using $x_j$ (for $1\leq j\leq m$),
\item $f_i  \rightarrow c_{j_1}\shuffle\ldots\shuffle c_{j_k}$, where $c_{j_1},\ldots,c_{j_k}$ correspond to the clauses using $\neg x_j$ (for $1\leq j\leq m$).
\end{itemize}
\item The twig query $q = r[\dblslash c_1]\dots[\dblslash c_n]$.
\end{itemize}
For example, for the formula $\varphi_0=(x_1\vee \neg x_2 \vee x_3)\wedge(\neg x_1\vee x_3\vee \neg x_4)$ we obtain the DIMS $S$ containing the rules:
\begin{gather*}
r \rightarrow (t_1\mid f_1) \shuffle (t_2\mid f_2) \shuffle (t_3\mid f_3) \shuffle (t_4\mid f_4), \\
t_1 \rightarrow c_1, \qquad f_1\rightarrow c_2, \qquad t_2\rightarrow \epsilon, \qquad f_2\rightarrow c_1,\\
t_3\rightarrow c_1\shuffle c_2, \qquad f_3 \rightarrow \epsilon, \qquad t_4 \rightarrow \epsilon, \qquad f_4 \rightarrow c_2.
\end{gather*}
and the query $q = /r[\dblslash c_1][\dblslash c_2]$.
The formula $\varphi$ is satisfiable iff $(S, q)\in \SAT_{\dims, \Twig}$.
The described reduction works in polynomial time in the size of the input formula.

\rnew{
For the NP upper bound, we reduce $\SAT_{\dims, \Twig}$ to $\SAT_{\mathit{DTD}, \Twig}$ (i.e., the problem of satisfiability of twig queries in the presence of DTDs), known to be in NP (Theorem 4.4 from~\cite{BeFaGe08}).
Given a DIMS $S$, we construct a DTD $D$ having the same root label as $S$ and whose rules are obtained from the rules of $S$ by replacing the unordered concatenation with standard (ordered) concatenation.
Then, take a twig query $q$.
We claim that there exists an (unordered) tree satisfying $q$ and $S$ iff there exists an (ordered) tree satisfying $q$ and $D$.
For the \emph{if} part, take an ordered tree $t$ satisfying $q$ and $D$, remove the order to obtain an unordered tree $t'$, and observe that $t'$ satisfies $S$.
For the \emph{only if} part, take an unordered tree $t$ satisfying $q$ and $S$.
From the construction of $D$, we infer that there exists an ordered tree $t'$ (obtained via some ordering of the sibling nodes of $t$) satisfying both $q$ and $D$.
We recall that the twig queries disregard the relative order among the siblings.
}
\qed\end{proof}
The complexity results for query implication and query containment in the presence of DIMSs follow from the EXPTIME-completeness proof from \cite{NeSc06} for twig query containment in the presence of DTDs.

\begin{proposition}\label{propexptimec}
$\IMPL_{\dims,\Twig}$ and $\CNT_{\dims,\Twig}$ are EXPTIME-complete.
\end{proposition}
\begin{proof}
The EXPTIME-hardness proof of twig containment in the presence of DTDs (Theorem 4.5 from~\cite{NeSc06}) has been done using a reduction from the \emph{Two-player corridor tiling} problem and a technique introduced in \cite{MiSu04}.
In the proof from \cite{NeSc06}, when testing the containment $p\subseteq_Sq$, $p$ is chosen such that it satisfies every tree in $S$, hence $\IMPL_{\DTD,\Twig}$ is also EXPTIME-complete.
Furthermore, Lemma 3 in \cite{MiSu04} can be adapted to twig queries and DIMS: for every $S\in\dims$ and twig queries $q_0,q_1,\ldots,q_m$ there exists $S'\in\dims$ and twig queries $q$ and $q'$ such that $q_0 \subseteq_S q_1\cup\ldots\cup q_m$  iff $q \subseteq_{S'} q'$.
Moreover, the DTD in \cite{NeSc06} can be \rnew{captured with a DIMS constructible in polynomial time: take the same reduction as in \cite{NeSc06} and then replace the standard concatenation with unordered concatenation.
Hence, we infer that $\CNT_{\dims,\Twig}$ and $\IMPL_{\dims,\Twig}$ are also EXPTIME-hard. 

For the EXPTIME upper bound, we reduce $\CNT_{\dims,\Twig}$ to $\CNT_{\mathit{DTD},\Twig}$ (i.e., the problem of twig query containment in the presence of DTDs), known to be in EXPTIME (Theorem 4.4 from~\cite{NeSc06}).
Given a DIMS $S$, we construct a DTD $D$ having the same root label as $S$ and whose rules are obtained from the rules of $S$ by replacing the unordered concatenation with standard (ordered) concatenation.
Then, take two twig queries $p$ and $q$.
We claim that $p\subseteq_S q$ iff $p\subseteq_D q$ and show the two parts by contraposition.
For the \emph{if} part, assume $p\not\subseteq_S q$, hence there exists an unordered tree $t$ that satisfies $q$ and $S$, but not $p$.
From the construction of $D$, we infer that there exists an ordered tree $t'$ (obtained via some ordering of the sibling nodes of $t$) that satisfies $q$ and $D$, but not $p$.
For the \emph{only if} part, assume $p\not\subseteq_D q$, hence there exists an ordered tree $t$ that satisfies $q$ and $D$, but not $p$.
By removing the order of $t$, we obtain an unordered tree $t'$ that satisfies $q$ and $S$, but not $p$.
We recall that the twig queries disregard the relative order among the siblings.
The membership of $\CNT_{\dims,\Twig}$ to EXPTIME yields that $\IMPL_{\dims,\Twig}$ is also in EXPTIME (it suffices to take as $p$ the universal query).
}
\qed\end{proof}

\section{Complexity of disjunction-free interval multiplicity schemas (IMSs)}\label{ms}
Although \emph{query satisfiability} and \emph{query implication} in the presence of schema are intractable for DIMSs, we prove that they become tractable for IMSs (Section~\ref{subsec:complex}).
We also show a considerably lower complexity for \emph{query containment} in the presence of schema: coNP-completeness for IMSs instead of EXPTIME-completeness for DIMSs (Section~\ref{subsec:complex}).
Additionally, we point out that our results for IMSs allow also to characterize the complexity of query implication and
query containment in the presence of \emph{disjunction-free DTDs} \rnew{(i.e., restricted DTDs using regular expressions without disjunction operator)}, which, to the best of our knowledge, have not been previously studied (Section~\ref{subsec:dtd}).
To prove our results, we develop a set of tools that we present next: \emph{dependency graphs} (Section~\ref{subsec:dep}), \emph{generalized definition of embedding} (Section~\ref{subsec:embed}), \emph{family of characteristic graphs} (Section~\ref{subsec:family}).

\subsection{Dependency graphs}\label{subsec:dep}
Recall that IMSs use IMEs, which are essentially expressions of the form $A_1^{I_1}\shuffle\ldots\shuffle A_k^{I_k}$, where $A_1,\ldots, A_k$ are atoms, and $I_1,\ldots,I_k$ are intervals.
Given an IME $E$, let $\simbu(E)$ be the set of symbols present in all unordered words in $L(E)$, and $\simbe(E)$ the set of symbols present in at least one unordered word in $L(E)$:
\begin{gather*}
\simbu(E) = \{a\in\Sigma\mid \forall w\in L(E).\ a\in w\},\\
\simbe(E) = \{a\in\Sigma\mid \exists w\in L(E).\ a\in w\}.
\end{gather*}
Given an IME $E$, notice that $\simbu(E)\subseteq\simbe(E)$, and moreover, the sets $\simbu(E)$ and $\simbe(E)$ can be easily constructed from $E$. 
For example, given $E_0 = (a\shuffle b^?)^{[5,6]}\shuffle c^+$, we have $\simbu(E_0) = \{a,c\}$ and $\simbe(E_0) = \{a,b,c\}$.

\begin{definition}
\rnew{Given an IMS $S=(\root_S, R_S)$, the \emph{existential dependency graph} of $S$ is the directed rooted graph $G_S^\exists = (\Sigma,\root_S, E_S^\exists)$ with the node set $\Sigma$, the distinguished root node $\root_S$, and the set of edges $E_S^\exists$ such that $(a,b)\in E_S^\exists$ if $b\in\simbe(R_S(a))$.
Furthermore, the \emph{universal dependency graph} of $S$ is the directed rooted graph $G_{S}^{\forall} = (\Sigma,\root_S, E_S^{\forall})$ such that $(a,b)\in E_S^{\forall}$ if $b\in\simbu(R_S(a))$.}
\end{definition}
\begin{example}\label{ex:graphs}\normalfont
Take the IMS $S$ containing the rules:
\begin{gather*}
r \rightarrow (a^?\shuffle b)^{[1,10]}\shuffle c, \qquad
a \rightarrow d^?, \qquad
b \rightarrow a^{[2,3]}\shuffle c^*\shuffle d^+.
\end{gather*}
In Figure~\ref{figgraphs} we present the existential dependency graph of $S$ and the universal dependency graph of $S$.
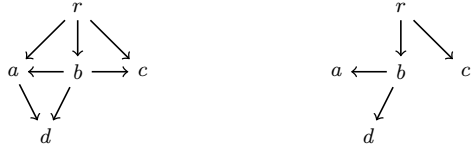
\begin{figure}[htb]
  \centering
  \begin{tikzpicture}[yscale=0.85, xscale=0.85]
    \begin{scope}
      \node at (0,0) (m99) {$r$};
      \node at (-1,-1) (m0) {$\mathit{a}$};
      \node at (0,-1) (m10) {$\mathit{b}$};
      \node at (1,-1) (m20) {$\mathit{c}$};
      \node at (-0.5,-2) (m3) {$\mathit{d}$};
      \draw[->,semithick] (m99) -- (m0);			
      \draw[->,semithick] (m99) -- (m10);
      \draw[->,semithick] (m99) -- (m20);
      \draw[->,semithick] (m10) -- (m0);
      \draw[->,semithick] (m10) -- (m20);
      \draw[->,semithick] (m10) -- (m3);
      \draw[->,semithick] (m0) -- (m3);
    \end{scope}
    \begin{scope}[xshift=5cm]
      \node at (0,0) (m99) {$r$};
      \node at (-1,-1) (m0) {$\mathit{a}$};
      \node at (0,-1) (m10) {$\mathit{b}$};
      \node at (1,-1) (m20) {$\mathit{c}$};
      \node at (-0.5,-2) (m3) {$\mathit{d}$};
      \draw[->,semithick] (m99) -- (m10);			
      \draw[->,semithick] (m99) -- (m20);
      \draw[->,semithick] (m10) -- (m0);
      \draw[->,semithick] (m10) -- (m3);
    \end{scope}
  \end{tikzpicture}
  \caption{\label{figgraphs}Existential dependency graph $G_{S}^\exists$ and universal dependency graph $G_{S}^{\forall}$ for Example~\ref{ex:graphs}.}
\end{figure}
\qed\end{example}
\noindent
Given an IMS $S$ and a symbol $a$, we say that $a$ is \emph{\rnew{reachable}} (or \rnew{\em useful}) in $S$ if there exists a tree in $L(S)$ which has a node labeled by $a$.
Moreover, we say that an IMS is \emph{\rnew{trimmed}} if it contains rules only for the reachable symbols.
For every satisfiable IMS $S$, there exists an equivalent trimmed version which can be obtained by removing the rules for the symbols involved in \rnew{\em unreachable components} in $G_S^{\forall}$ (\rnew{in the spirit of~\cite{AlGiWo01}}).
Notice that the unreachable components of $G_S^{\forall}$ correspond in fact to cycles in $G_S^{\forall}$.
In the sequel, we assume w.l.o.g.\ that all IMSs that we manipulate are satisfiable and trimmed.

\subsection{Generalizing the embedding}\label{subsec:embed}
We generalize the notion of embedding previously defined in Section~\ref{sec:prelim}.
Note that in the rest of the section we use the term \rnew{\emph{dependency graphs}} when we refer to both existential and universal dependency graphs.
First, an \emph{embedding} of a query $q$ in a dependency graph $G=(\Sigma, \root,E)$ is a function $\lambda : N_q \rightarrow \Sigma$ such that:
\begin{enumerate}
\itemsep0pt
\item[$1$.] $\lambda(\root_q)=\root$,
\item[$2$.] for every $(n,n')\in \child_q$, $(\lambda(n),\lambda(n'))\in E$,
\item[$3$.] for every $(n,n')\in \desc_q$, $(\lambda(n),\lambda(n'))\in E^+$ (the transitive closure of $E$),
\item[$4$.] for every $n\in N_q$, $\lab_q(n) = \wc$ or $\lab_q(n) = \lambda(n)$.
\end{enumerate}
If there exists an embedding of $q$ in $G$, we write $G\preccurlyeq q$.
Next, a \emph{simulation} of a dependency graph $G=(\Sigma,\root,E)$ in a tree $t$ is a relation $R\subseteq \Sigma\times N_t$ such that: 
\begin{enumerate}
\item $(\root, \root_t)\in R$,
\item for every $(a, n)\in R,\ (a, a')\in E$, there exists $n'\in N_t$ such that $(n,n')\in\child_t$ and $(a',n')\in R$,
\item for every $(a, n)\in R.\ \lab_t(n) = a$.
\end{enumerate}
\noindent
Note that $R$ is a total relation for the nodes of the graph reachable from the root i.e., for every $a\in \Sigma$ reachable from $\root$ in $G$, there exists a node $n\in N_t$ such that $(a,n)\in R$. 
If there exists a simulation from $G$ to $t$, we write $t\preccurlyeq G$.
\rnew{
Additionally, note that given a graph containing cycles reachable from the root, there does not exist any (finite) tree where it can be simulated.
However, we point out that in the remainder we use the notion of simulation only for universal dependency graphs that are supposed to come from trimmed IMSs, hence they do not have such cycles.}

Given two dependency graphs $G_1=(\Sigma,\root,E_1)$ and $G_2=(\Sigma,\root,E_2)$, $G_1$ is a \emph{subgraph} of $G_2$ if $E_1\subseteq E_2$.
For a dependency graph $G=(\Sigma,\root,E)$, we define the partial order $\leq_G$ on the subgraphs of $G$: given $G_1$ and $G_2$ two subgraphs of $G$, $G_1\leq_G G_2$ if $G_1$ is a subgraph of $G_2$.
Note that the relation $\leq_G$ is reflexive, antisymmetric, and transitive, thus being an ordering relation.
Moreover, it is well-founded and it has a minimal element $G_0 = (\Sigma,\root, \emptyset)$.
The following result can be easily shown by a structural induction using the order $\leq_G$.
\begin{lemma}\label{L1}
For every IMS $S$, its universal dependency graph can be simulated in every tree $t$ which belongs to the language of $S$.
\end{lemma}
\noindent
A \rnew{\emph{path in a dependency graph}} $G =(\Sigma, \root, E)$ is a non-empty sequence of vertices starting at $\root$ such that for every two consecutive vertices in the sequence, there is a directed edge between them in $G$. 
By $\Paths(G)\subseteq\Sigma^+$ we denote the set of all paths in $G$. 
The set of paths is finite only for graphs without cycles reachable from the root. 
For instance, the paths of the graph $G_1$ in Figure~\ref{fig:unfolding} are $\Paths(G_1) = \{r, ra, rb, rc, rbd, rcd, rbde, rcde\}$. 

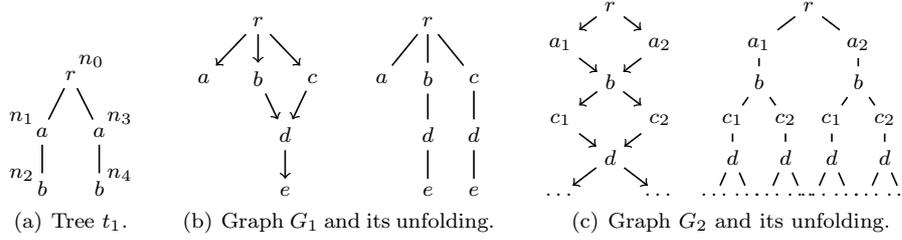
\begin{figure}[htb]
	\centering
	\subfigure[Tree $t_1$.]{\label{fig:LabPaths}
  \begin{tikzpicture}[yscale=0.75,xscale=0.75]
    \begin{scope}
      \node at (0,1) (m99) {$r$};
      \node at (0.5,0) (m0) {$a$};
      \node at (-0.5,0) (m1) {$a$};
      \node at (0.5,-1) (m2) {$b$};
      \node at (-0.5,-1) (m23) {$b$};
      \draw[-,semithick] (m99) -- (m0);	
      \draw[-,semithick] (m99) -- (m1);		
      \draw[-,semithick] (m0) -- (m2);
      \draw[-,semithick] (m1) -- (m23);
      \path (m99) node[above right] {$n_0$};
      \path (m1) node[above left] {$n_1$};
      \path (m23) node[above left] {$n_2$};
      \path (m2) node[above right] {$n_4$};
      \path (m0) node[above right] {$n_3$};
    \end{scope}
  \end{tikzpicture}
	}~~~~~~
\subfigure[Graph $G_1$ and its unfolding.]{\label{fig:unfolding}
  \begin{tikzpicture}[yscale=0.75,xscale=0.95]
    \node at (0,0) (n0) {$r$};
    \node at (-0.75,-1) (n1) {$a$};
    \node at (0,-1) (n2) {$b$};
    \node at (0.75,-1) (n3) {$c$};
    \node at (0.375,-2) (n4) {$d$};
    \node at (0.375,-3) (n5) {$e$};
    \draw[->,semithick] (n0) -- (n1);
    \draw[->,semithick] (n0) -- (n2);
    \draw[->,semithick] (n0) -- (n3);
    \draw[->,semithick] (n2) -- (n4);
    \draw[->,semithick] (n3) -- (n4);
    \draw[->,semithick] (n4) -- (n5);
    \begin{scope}[xshift=2.35cm,xscale=0.85]
    \node at (0,0) (n0) {$r$};
    \node at (-0.75,-1) (n1) {$a$};
    \node at (0,-1) (n2) {$b$};
    \node at (0.75,-1) (n3) {$c$};
    \node at (0,-2) (n4) {$d$};
    \node at (0.75,-2) (n5) {$d$};
    \node at (0,-3) (n6) {$e$};
    \node at (0.75,-3) (n7) {$e$};
    \draw[-,semithick] (n0) -- (n1);
    \draw[-,semithick] (n0) -- (n2);
    \draw[-,semithick] (n0) -- (n3);
    \draw[-,semithick] (n2) -- (n4);
    \draw[-,semithick] (n3) -- (n5);
    \draw[-,semithick] (n4) -- (n6);
    \draw[-,semithick] (n5) -- (n7);
    \end{scope}
  \end{tikzpicture}
}~~~~~~
\subfigure[Graph $G_2$ and its unfolding.]{  \label{fig:exponential unfolding}
  \begin{tikzpicture}[yscale=0.5,xscale=0.87]
    \node at (0,0) (n0) {$r$};
    \node at (-0.75,-1) (n1) {$a_1$};
    \node at (0.75,-1) (n2) {$a_2$};
    \node at (0,-2) (n3) {$b$};
    \node at (-0.75,-3) (n4) {$c_1$};
    \node at (0.75,-3) (n5) {$c_2$};
    \node at (0,-4) (n6) {$d$};
    \node at (-0.75,-5) (n7) {$\dots$};
    \node at (0.75,-5) (n8) {$\dots$};
    \draw[->,semithick] (n0) -- (n1);
    \draw[->,semithick] (n0) -- (n2);
    \draw[->,semithick] (n1) -- (n3);
    \draw[->,semithick] (n2) -- (n3);
    \draw[->,semithick] (n3) -- (n4);
    \draw[->,semithick] (n3) -- (n5);
    \draw[->,semithick] (n4) -- (n6);
    \draw[->,semithick] (n5) -- (n6);
    \draw[->,semithick] (n6) -- (n7);
    \draw[->,semithick] (n6) -- (n8);
    \begin{scope}[xshift=3cm]
    \node at (0,0) (n0) {$r$};
    \node at (-0.75,-1) (n1) {$a_1$};
    \node at (0.75,-1) (n2) {$a_2$};
    \node at (-0.75,-2) (n3) {$b$};
    \node at (0.75,-2) (n30) {$b$};
    \node at (-1.15,-3) (n4) {$c_1$};
    \node at (-0.35,-3) (n5) {$c_2$};
    \node at (1.15,-3) (n40) {$c_2$};
    \node at (0.35,-3) (n50) {$c_1$};

    \node at (-1.15,-4) (n6) {$d$};
    \node at (-0.35,-4) (n60) {$d$};
    \node at (0.35,-4) (n61) {$d$};
    \node at (1.15,-4) (n62) {$d$};

    \node at (-1.4,-5) (n7) {$\dots$};    
   \node at (-0.9,-5) (n70) {$\dots$};

    \node at (-0.6,-5) (n71) {$\dots$};
    \node at (-0.1,-5) (n72) {$\dots$};

    \node at (0.1,-5) (n73) {$\dots$};
    \node at (0.6,-5) (n74) {$\dots$};

    \node at (0.9,-5) (n75) {$\dots$};
    \node at (1.4,-5) (n76) {$\dots$};
    \draw[-,semithick] (n0) -- (n1);
    \draw[-,semithick] (n0) -- (n2);
    \draw[-,semithick] (n1) -- (n3);
    \draw[-,semithick] (n2) -- (n30);
    \draw[-,semithick] (n3) -- (n4);
    \draw[-,semithick] (n3) -- (n5);
    \draw[-,semithick] (n30) -- (n40);
    \draw[-,semithick] (n30) -- (n50);
    \draw[-,semithick] (n60) -- (n5);
    \draw[-,semithick] (n6) -- (n4);
    \draw[-,semithick] (n61) -- (n50);
    \draw[-,semithick] (n62) -- (n40);
    \draw[-,semithick] (n6) -- (n7);
    \draw[-,semithick] (n6) -- (n70);
    \draw[-,semithick] (n60) -- (n71);
    \draw[-,semithick] (n60) -- (n72);
    \draw[-,semithick] (n61) -- (n73);
    \draw[-,semithick] (n61) -- (n74);
    \draw[-,semithick] (n62) -- (n75);
    \draw[-,semithick] (n62) -- (n76);
    \end{scope}
  \end{tikzpicture}
}
\caption{A tree and two graphs with their corresponding unfoldings.}
\end{figure}
\noindent
Similarly, a \emph{path in a tree} $t$ is a non-empty sequence of nodes starting at $\root_t$ such that every two consecutive nodes in the sequence are in the relation $child_t$. 
By $\Paths(t)\subseteq N_t^+$ we denote the set of all paths in $t$. 
Then, we define $\LabPaths(t)\subseteq \Sigma^+$ as the set of sequences of labels of nodes from all paths in $t$. 
For instance, for the tree $t_1$ from Figure~\ref{fig:LabPaths} we have $\Paths(t_1)=\{n_0,n_0n_1,n_0n_1n_2,n_0n_3,n_0n_3n_4\}$ and
$\LabPaths(t_1)=\{r$, $ra,rab\}$. Note that $|\LabPaths(t)|\leq|\Paths(t)|$.
The \rnew{\emph{unfolding} of a dependency graph} $G=(\Sigma,\root,E)$, denoted $u_G$, is a tree $u_G = (N_{u_G}, \root_{u_G}, \lab_{u_G}, \child_{u_G})$ such that:
\begin{itemize}
\item $N_{u_G} = \Paths(G)$,
\item $\root_{u_G}\in N_{u_G}$ is the root of $u_G$,
\item $(p, p\cdot a)\in\child_{u_G}$, for every path $p, p\cdot a\in \Paths(G)$ (note that ``$\cdot$'' stands for standard ordered concatenation),
\item $\lab_{u_G}(\root_{u_G}) = \root$, and $\lab_{u_G}(p\cdot a) = a$, for every path $p\cdot a\in\Paths(G)$. 
\end{itemize}
\noindent
The unfolding of a graph is finite only when the graph has no cycle reachable from the root, because otherwise $\Paths(G)$ is infinite, hence $u_G$ is infinite.
\rnew{In the remainder, we use the unfolding only for graphs having no cycle reachable from the root (in order to have finite unfoldings).}
In such a case, the unfolding can be seen as the \emph{smallest} tree ${u_G}$ (w.r.t.\ the number of nodes) having $\LabPaths({u_G}) = \Paths(G)$.
The idea of the unfolding is to transform the dependency graph $G$ into a tree having the $\child$ relation instead of directed edges. 
There are nodes duplicated in order to avoid nodes with more than one incoming edge.
For instance, in Figure~\ref{fig:unfolding} we take the graph $G_1$ and construct its unfolding $u_{G_1}$. 
Moreover, notice that the size of the unfolding may be exponential in the size of the graph, for example for the graph $G_2$ from Figure \ref{fig:exponential unfolding}.

We also extend the definition of embedding and propose the embedding from a tree to another tree i.e., given two trees $t$ and $t'$, we say that $t'$ can be embedded in $t$ (denoted $t\preccurlyeq t'$) if the query $(N_{t'},\root_{t'},\lab_{t'},\child_{t'},\emptyset)$ can be embedded in $t$.
Similarly, we can define the embedding from a tree to a dependency graph.
Note that two embeddings can be \emph{composed}, for example:
\begin{itemize}
\item $\forall t,t'\in\Tree.\ \forall q\in\Twig.\ (t\preccurlyeq t'\wedge t'\preccurlyeq q\Rightarrow t\preccurlyeq q)$,
\item $\forall S\in\sims.\ \forall t\in\Tree.\ \forall q\in\Twig.\ (G_S^{\forall/\exists}\preccurlyeq t\wedge t\preccurlyeq q\Rightarrow G_S^{\forall/\exists}\preccurlyeq q)$.
\end{itemize}
We state next two auxiliary lemmas that can be easily proven by structural induction on the dependency graphs (using the order $\leq_G$):
\noindent
\begin{lemma}\label{L2}
A \rnew{dependency graph} $G$ can be simulated in a tree $t$ iff its unfolding $u_G$ can be embedded in $t$.
\end{lemma}

\begin{lemma}\label{L3}
A query $q$ can be embedded in a \rnew{dependency graph} $G$ iff $q$ can be embedded in the unfolding tree of $G$.
\end{lemma}
In Figure \ref{fuseadd} we present the operations \emph{fuse} and \emph{add}. 
Given two trees $t$ and $t'$, we say that $t\lhd_0 t'$ if $t'$ is obtained from $t$ by applying one of the operations from Figure \ref{fuseadd}.
The \emph{fuse} operation takes two siblings with the same label and creates only one node having below it the subtrees corresponding to each of the siblings.
The \emph{add} operation consists simply in adding a subtree at some place in the tree.
By $\unlhd$ we denote the transitive and reflexive closure of $\lhd_0$.

\begin{figure}[htb] 
  \centering
  \begin{tikzpicture}[yscale=0.75,xscale=0.8]
   \node at (0.375,0) (n0) {$.$};
   \node at (-0.75,-1) (n1) {$a$};
   \node at (0,-1) (n2) {$b$};
   \node at (0.75,-1) (n3) {$b$};
   \node at (1.5,-1) (n4) {$c$};
   \draw[-,semithick] (n0) -- (n1);
   \draw[-,semithick] (n0) -- (n2);
   \draw[-,semithick] (n0) -- (n3);
   \draw[-,semithick] (n0) -- (n4);
   \draw (-0.75,-1.2) -- (-1,-2) -- (-0.5,-2)-- (-0.75,-1.2);
   \node at (-0.7,-1.75) (t1) {\scriptsize $t_1$};
   \draw (0,-1.2) -- (-0.25,-2) -- (0.25,-2)-- (0,-1.2);
   \node at (0.05,-1.75) (t2) {\scriptsize $t_2$};
   \draw (0.75,-1.2) -- (0.5,-2) -- (1,-2)-- (0.75,-1.2);
   \node at (0.8,-1.75) (t3) {\scriptsize $t_3$};
   \draw (1.5,-1.2) -- (1.25,-2) -- (1.75,-2)-- (1.5,-1.2);
   \node at (1.55,-1.75) (t3) {\scriptsize $t_4$};
  \begin{scope}[xshift=2cm,yshift=1cm]
   \node at (0,-1) (n0) {$\underrightarrow{\mathit{fuse}}$};
  \end{scope}
  \begin{scope}[xshift=3.5cm,xscale=0.95]
   \node at (0.375,0) (n0) {$.$};
   \node at (-0.75,-1) (n1) {$a$};
   \node at (0.375,-1) (n2) {$b$};
   \node at (1.5,-1) (n4) {$c$};
   \draw[-,semithick] (n0) -- (n1);
   \draw[-,semithick] (n0) -- (n2);
   \draw[-,semithick] (n0) -- (n4);
   \draw (-0.75,-1.2) -- (-1,-2) -- (-0.5,-2)-- (-0.75,-1.2);
   \node at (-0.7,-1.75) (t1) {\scriptsize $t_1$};
   \draw (0.36,-1.2) -- (-0.25,-2) -- (0.36,-2)-- (0.36,-1.2);
   \node at (0.2,-1.75) (t2) {\scriptsize $t_2$};
   \draw (0.39,-1.2) -- (0.39,-2) -- (1,-2)-- (0.39,-1.2);
   \node at (0.65,-1.75) (t3) {\scriptsize $t_3$};
   \draw (1.5,-1.2) -- (1.25,-2) -- (1.75,-2)-- (1.5,-1.2);
   \node at (1.55,-1.75) (t3) {\scriptsize $t_4$};
  \end{scope}
  \begin{scope}[xshift=8cm,xscale=1]
    \node at (0,0) (n0) {$.$};
    \node at (-0.75,-1) (n1) {$a$};
    \node at (0,-1) (n2) {$b$};
    \node at (0.75,-1) (n3) {$c$};
    \draw[-,semithick] (n0) -- (n1);
    \draw[-,semithick] (n0) -- (n2);
    \draw[-,semithick] (n0) -- (n3);
    \draw (-0.75,-1.2) -- (-1,-2) -- (-0.5,-2)-- (-0.75,-1.2);
    \node at (-0.7,-1.75) (t1) {\scriptsize $t_1$};
    \draw (0,-1.2) -- (-0.25,-2) -- (0.25,-2)-- (0,-1.2);
    \node at (0.05,-1.75) (t2) {\scriptsize $t_2$};
    \draw (0.75,-1.2) -- (0.5,-2) -- (1,-2)-- (0.75,-1.2);
    \node at (0.8,-1.75) (t3) {\scriptsize $t_3$};
  \end{scope}
  \begin{scope}[xshift=9.5cm,yshift=1cm]
   \node at (0,-1) (n0) {$\underrightarrow{\mathit{add}}$};
  \end{scope}
  \begin{scope}[xshift=11cm,xscale=1]
    \node at (0.375,0) (n0) {$.$};
    \node at (-0.75,-1) (n1) {$a$};
    \node at (0,-1) (n2) {$b$};
    \node at (0.75,-1) (n3) {$c$};
    \node at (1.5,-1) (n4) {$d$};
    \draw[-,semithick] (n0) -- (n1);
    \draw[-,semithick] (n0) -- (n2);
    \draw[-,semithick] (n0) -- (n3);
    \draw[-,semithick] (n0) -- (n4);
   \draw (-0.75,-1.2) -- (-1,-2) -- (-0.5,-2)-- (-0.75,-1.2);
   \node at (-0.7,-1.75) (t1) {\scriptsize $t_1$};
   \draw (0,-1.2) -- (-0.25,-2) -- (0.25,-2)-- (0,-1.2);
   \node at (0.05,-1.75) (t2) {\scriptsize $t_2$};
   \draw (0.75,-1.2) -- (0.5,-2) -- (1,-2)-- (0.75,-1.2);
   \node at (0.8,-1.75) (t3) {\scriptsize $t_3$};
   \draw (1.5,-1.2) -- (1.25,-2) -- (1.75,-2)-- (1.5,-1.2);
   \node at (1.55,-1.75) (t3) {\scriptsize $t_4$};
  \end{scope}
 \end{tikzpicture}
 \caption{\label{fuseadd}Operations \emph{fuse} and \emph{add}.}
\end{figure}
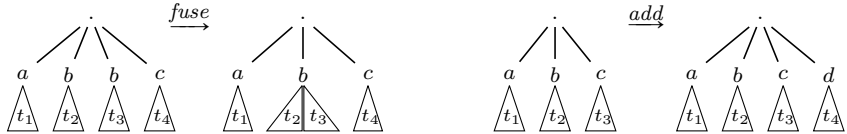
\noindent Note that the fuse and add operations preserve the embedding i.e., given a twig query $q$ and two trees $t$ and $t'$, if $t\preccurlyeq q$ and $t\unlhd t'$, then $t'\preccurlyeq q$.
Furthermore, if we can embed a query $q$ in a tree $t$ which can be embedded in the existential dependency graph of an IMS $S$, we can perform a sequence of operations such that $t$ is transformed into another tree $t'$ satisfying $S$ and $q$ at the same time. 
\rnew{Formally, we have the following.}
\begin{lemma}\label{L6}
Given an IMS $S$, a query $q$ and a tree $t$, if $G_S^\exists\preccurlyeq t$ and $t\preccurlyeq q$, then there exists a tree $t'\in L (S)\cap L(q)$.
The tree $t'$ can be constructed after a sequence of fuse and add operations (consistently with the schema $S$) from the tree $t$ and we denote $t\unlhd_S t'$.
\end{lemma}

\subsection{Family of characteristic graphs}\label{subsec:family}
Given a schema $S$ and a query $q$, we can capture all trees satisfying both $S$ and $q$ with the characteristic graphs that we introduce next.

\rnew{
More formally, a \emph{characteristic graph $G$} is a tuple $(V_G,\root_G,\lab_G,E_G)$, where $V_G$ is a finite set of vertices, $\root_G\in V_G$ is the root of the graph, $\lab_G:V_G\rightarrow\Sigma$ is a labeling function (with $\lab_G(\root_G) = \root_S$), and $E_G\subseteq V_G\times V_G$ is the set of edges.
Let us assume that $G_S^\exists\preccurlyeq q$ and take such an embedding $\lambda:N_q\rightarrow\Sigma$.
By $\Lambda(q,S,\lambda)$ we denote the set of all \emph{characteristic graphs for $q$ and $S$ w.r.t.\ $\lambda$}.
To construct such a graph, let us start with $G = (V_G,\root_G,\lab_G,E_G)$ where $V_G$ and $E_G$ are empty, and perform the four steps described below.
\begin{enumerate}
\item For \rnew{every} $n$ in $N_q$, add a node $n'$ to $V_G$ such that $\lab_G(n')=\lambda(n)$.
Let $\root_G$ be the node such that $\lab_G(\root_G)=\root_S$.

\item For \rnew{every} $(n_1,n_2)$ in $\child_q$, add $(n_1', n_2')$ to $E_G$, where $n_1'$ and $n_2'$ are the nodes corresponding to $n_1$ and $n_2$, respectively, as constructed at step 1.

\item For \rnew{every} $(n_1, n_2)$ in $\desc_q$, \emph{choose an acyclic path} $a_0,\ldots,a_k$ in $G_S^\exists$ where $\lambda(n_1)=a_0$ and $\lambda(n_2) = a_k$.
Notice that, since $n_1$ and $n_2$ belong to $N_q$, we have already added in $V_G$ two nodes $n_1'$ and $n_2'$, respectively, corresponding to them at step 1.
Then, for every $a_i$ (with $1\leq i\leq k-1$), we add in $V_G$ a node $n_i''$ such that $\lab_G(n_i'') = a_i$.
Also, add in $E_G$ the edges $(n_1',n_1''), (n_1'',n_2''),\ldots, (n_{k-1}'',n_2')$.

\item For \rnew{every} $n$ in $V_G$, take from $G_S^\forall$ the subgraph $(V',\lab_G(n),E')$ rooted at $\lab_G(n)$.
Then, for every $a\neq\lab_G(n)$ in $V'$ add a node $n'$ in $V'$ such that $\lab_G(n')=a$.
Also, for every $(a_1,a_2)\in E'$, add in $E_G$ an edge $(n_1,n_2)$ where $n_1$ and $n_2$ are the nodes corresponding to $a_1$ and $a_2$, respectively.
\end{enumerate}
}
\noindent The following example illustrates the construction of such a graph.

\begin{example}
Take in Figure \ref{fig:lambda} an existential dependency graph $G_S^\exists$, a twig query $q$, and an embedding $\lambda:N_q\rightarrow G_S^\exists$.
Notice that in $G_S^\exists$ we have drawn the universal edges with a full line and those that are existential without being universal with a dotted line.
Then, in Figure \ref{fig:graph} we present an example \rnew{of a graph} $G$ from $\Lambda(q,S,\lambda)$.
Notice that in $G$ we have represented in boxes the nodes corresponding to the images $\lambda(n)$ for the nodes of the query $n\in N_q$.
\qed
\end{example}
\begin{figure}[htb]
\centering
\subfigure[Embedding $\lambda:N_q\rightarrow G_S^\exists$.]{\label{fig:lambda}
\centering
\begin{tikzpicture}
\node at (-2,0) (s0) {$r$};
\node at (-2,-1) (s1) {$c$}edge[<-] (s0);
\node at (-3,-2) (s2) {$a_1$}edge[<-] (s1);
\node at (-1,-2) (s3) {$a_2$}edge[<-] (s1);
\node at (-2,-3) (s4) {$b$}edge[<-] (s2)edge[<-] (s3)edge[->,densely dotted] (s1);;
\node at (1,0) (n0) {$r$}edge[->,bend right, densely dotted] (s0);
\node at (0.5,-1) (n1) {$\wc$}edge[-,double] (n0)edge[->,bend right, densely dotted] (s3);
\node at (0.5,-2) (n2) {$\wc$}edge[-,double] (n1)edge[->,bend right, densely dotted] (s4);
\node at (0.5,-3) (n3) {$\wc$}edge[-,double] (n2)edge[->,bend left, densely dotted] (s3);
\node at (1.5,-1) (n4) {$b$}edge[-,double] (n0)edge[->,bend left, densely dotted] (s4);
\node at (1.5,-2) (n5) {$c$}edge[-] (n4)edge[->,bend right, densely dotted] (s1);
\end{tikzpicture}
}
\subfigure[Graph $G\in\Lambda(q,S,\lambda)$.]{\label{fig:graph}~~~~~~~~~~~~~
\centering
\begin{tikzpicture}
\node at (0,0) (g0) {$\framebox{$r$}$};
\node at (1,0) (g00) {$c$} edge[<-] (g0);
\node at (1.9,0.25) (g01) {$a_1$} edge [<-] (g00);
\node at (1.9,-0.25) (g02) {$a_2$} edge [<-] (g00);
\node at (2.8,0) (g03) {$b$} edge [<-] (g01)edge [<-] (g02);
\node at (-2,-1) (g1) {$c$} edge[<-] (g0);
\node at (-1,-0.75) (g10) {$a_1$} edge [<-] (g1);
\node at (-1,-1.25) (g11) {$a_2$} edge [<-] (g1);\node at (-0.1,-1) (g12) {$b$} edge [<-] (g10)edge [<-] (g11);
\node at (-2,-2) (g2) {$\framebox{$a_2$}$} edge[<-] (g1);
\node at (-1.1,-2) (g20) {$b$} edge[<-] (g2);
\node at (-2,-3) (g3) {$\framebox{$b$}$} edge[<-] (g2);
\node at (-2,-4) (g4) {$c$} edge[<-] (g3);
\node at (-1,-3.75) (g40) {$a_1$} edge [<-] (g4);
\node at (-1,-4.25) (g41) {$a_2$} edge [<-] (g4);
\node at (-0.1,-4) (g42) {$b$} edge [<-] (g40)edge [<-] (g41);
\node at (-2,-5) (g5) {$\framebox{$a_2$}$} edge[<-] (g4);
\node at (-1,-5) (g50) {$b$} edge[<-] (g5);
\node at (1,-1) (g6) {$c$} edge[<-] (g0);
\node at (2,-0.75) (g60) {$a_1$} edge [<-] (g6);
\node at (2,-1.25) (g61) {$a_2$} edge [<-] (g6);
\node at (2.9,-1) (g62) {$b$} edge [<-] (g60)edge [<-] (g61);
\node at (1,-2) (g7) {$a_1$} edge[<-] (g6);
\node at (1.9,-2) (g70) {$b$} edge[<-] (g7);
\node at (1,-3) (g8) {$\framebox{$b$}$} edge[<-] (g7);
\node at (1,-4) (g9) {$\framebox{$c$}$} edge[<-] (g8);
\node at (2,-3.75) (g90) {$a_1$} edge [<-] (g9);
\node at (2,-4.25) (g91) {$a_2$} edge [<-] (g9);
\node at (2.9,-4) (g92) {$b$} edge [<-] (g90)edge [<-] (g91);
\end{tikzpicture}
}
\caption{\label{conp-alg}An embedding from a query $q$ to an existential dependency graph $G_S^\exists$ and a graph $G\in \mathcal G(q,S)$.
In $G_S^\exists$, the universal edges are drawn with a full line and those that are existential without being universal with a dotted line.}
\end{figure}
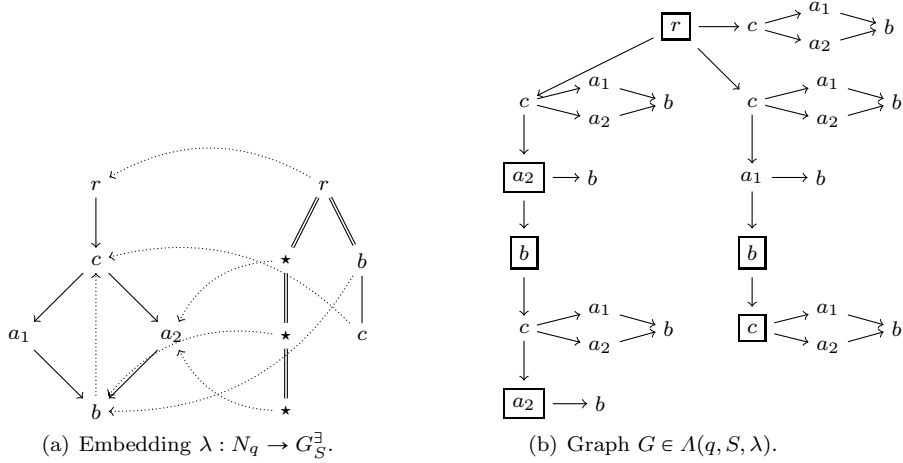
Next, we define the set of all characteristic graphs for $q$ and $S$ w.r.t.\ the all embeddings $\lambda$ of $q$ in $G_S^\exists$:

\[
\mathcal G(q,S)=\{G\in\Lambda(q,S,\lambda)\mid \lambda \text{ is an embedding of } q \text{ in } G_S^\exists\}.
\]

\noindent
\rnew{Note that $G\preccurlyeq q$ and the size of $G$ is polynomially bounded by $|q|\times |\Sigma|^2$ for every $G$ in $\mathcal G(q,S)$.
Indeed, after step 1 of the construction, a characteristic graph $G$ has $|q|$ nodes.
Then, after steps 2 and 3, since at step 3 we allow only acyclic paths of $G_S^\exists$, we add at most $|\Sigma|$ nodes for each already existing node, hence $G$ has at most $|q|\times|\Sigma|$ nodes.
Finally, after 4, since we add at most $|\Sigma|$ nodes for each already existing node, $G$ has at most $|q|\times|\Sigma|^2$ nodes.}

Furthermore, let $\Lambda^*(q,S,\lambda)$ and $\mathcal G^*(q,S)$ be sets of characteristic graphs constructed similarly to $\Lambda(q,S,\lambda)$ and $\mathcal G(q,S)$, respectively, the only difference being that we allow cyclic paths at step 3 of the aforementioned construction.
While the size of the graphs in $\mathcal G(q,S)$ is polynomial , notice that the size of the graphs in $\mathcal G^*(q,S)$ is not necessary polynomial since the possible cyclic paths chosen at step 3 can be arbitrarily long.
Additionally, note that $|\mathcal G(q,S)|$ is finite and may be exponential while $|\mathcal G^*(q,S)|$ may be infinite if the existential dependency graph $G_S^\exists$ contains cycles reachable from the root.

Next, we extend the previous definition of the unfolding to the characteristic graphs.
Given an IME $E$ and a symbol $a$, by $\minnb(E,a)$ we denote the minimum number of occurrences of the symbol $a$ in every unordered word defined by $E$.
Next, we define the \emph{unfolding of a characteristic graph}.
Given a query $q$, an IMS $S$, and a characteristic graph $G\in\mathcal G^*(q,S)$, we construct its unfolding as follows:

\begin{itemize}
\item Let $u_{G}$ be the unfolding of $G$ obtained as defined in Section~\ref{subsec:embed}\rnew{.}
\item Update $u_{G}$ such that for every $n\in N_{u_{G}}$, for every $a\in\Sigma$, 
let $t_a$ the subtree having as root the child of $n$ labeled by $a$.
Next, add copies of $t_a$ as children of $n$ until $n$ has $\minnb(R_S(\lab_{u_{G}}(n)), a)$ children labeled by $a$.
\end{itemize}

\noindent
\rnew{
Notice that every graph $G$ in $\mathcal G^*(q,S)$ is acyclic.
Indeed, when constructing such a graph $G$, after steps 1, 2 and 3, $G$ is basically shaped as a tree.
Then, the subgraphs that we fuse at step 4 are all acyclic since they are subgraphs of the universal dependency graph $G_S^\forall$ that we assume trimmed (cf.\ Section~\ref{subsec:dep}).
Since every graph $G$ in $\mathcal G^*(q,S)$ is acyclic, it has a finite unfolding, which naturally belongs to the language of $S$.
}

\subsection{Complexity results}\label{subsec:complex}
In this section, we use the above defined tools to show the complexity results for IMSs.
First, the dependency graphs and embeddings capture satisfiability and
implication of queries by IMSs.
\begin{lemma}\label{lemmagraph}
Given a twig query $q$ and an IMS $S$: 
\begin{enumerate} 
\item $q$ is satisfiable by $S$ iff $G_S^\exists\preccurlyeq q$, 
\item $q$ is implied by $S$ iff $G_S^\forall\preccurlyeq q$.
\end{enumerate}
\end{lemma}
\begin{proof}
\emph{1)} For the \emph{if} part, we know that $G_S^\exists\preccurlyeq q$, thus the family of graphs $\mathcal G(q, S)$ is not empty.
The unfolding of every graph from $\mathcal G(q, S)$ satisfies $S$ and $q$ at the same time, hence $q$ is satisfiable by $S$.
For the \emph{only if} part, we know that there exists a tree $t\in L(S)\cap L(q)$, and we assume w.l.o.g.\ that it is the unfolding of a graph $G$ from $\mathcal G^*(q,S)$.
Since $t\preccurlyeq q$, we obtain $u_{G}\preccurlyeq q$, hence $G\preccurlyeq q$ (by Lemma~\ref{L3}), which, from the construction of $G$, implies that $G_S^\exists\preccurlyeq q$.

\emph{2)} For the \emph{if} part, we know that $G_S^{\forall}\preccurlyeq q$, which implies by Lemma~\ref{L3} that $u_{G_S^{\forall}}\preccurlyeq q$.
On the other hand, take a tree $t\in L(S)$.
By Lemma~\ref{L1} we have $t\preccurlyeq G_{S}^{\forall}$, which implies by Lemma~\ref{L2} that $t\preccurlyeq u_{G_S^{\forall}}$.
From the last embedding and $u_{G_S^{\forall}}\preccurlyeq q$ we infer that $t\preccurlyeq q$.
Since $t$ can be every tree in the language of $S$, we conclude that $q$ is implied by $S$.
For the \emph{only if} part, we know that for every $t\in L(S)$, $t\preccurlyeq q$.
Consider the tree $t$ obtained as follows: we take $u_{G_S^{\forall}}$ and we duplicate some subtrees in order to have, for each node $n\in N_t$, $\minnb(R_S(\lab_t(n)),a)$ children labeled by $a$.
Naturally, $t$ is in the language of $S$, hence $t\preccurlyeq q$ from the hypothesis.
From the definition of the unfolding, we infer that $G_S^{\forall}\preccurlyeq t$, which implies that $G_S^{\forall}\preccurlyeq q$.
\qed\end{proof}
For instance, the twig query $q=r[a]/b\dblslash d$ can be embedded in the existential dependency graph of the IMS $S$ from Example~\ref{ex:graphs}, thus $q$ is satisfiable by $S$.
In Figure~\ref{fig:implication} we present embeddings of $q$ in $G_S^\exists$ and in a tree $t$ satisfying both $S$ and $q$.
Additionally, notice that the twig query $q=r[a]/b\dblslash d$ cannot be embedded in $G_S^{\forall}$ from Example~\ref{ex:graphs}, and therefore, $q$ is not implied by $S$.
On the other hand, the twig query $q'= r/b\dblslash d$ can be embedded in $G_S^{\forall}$, thus $q'$ is implied by $S$.

\begin{figure}[htb]
\center
  \begin{tikzpicture}[yscale=0.85, xscale=0.85]
    \begin{scope}
      \node at (0,0) (m99) {$r$};
      \node at (-1,-1) (m0) {$\mathit{a}$};
      \node at (0,-1) (m10) {$\mathit{b}$};
      \node at (1,-1) (m20) {$\mathit{c}$};
      \node at (-0.5,-2) (m3) {$\mathit{d}$};
      \draw[->,semithick] (m99) -- (m0);			
      \draw[->,semithick] (m99) -- (m10);
      \draw[->,semithick] (m99) -- (m20);
      \draw[->,semithick] (m10) -- (m0);
      \draw[->,semithick] (m10) -- (m20);
      \draw[->,semithick] (m10) -- (m3);
      \draw[->,semithick] (m0) -- (m3);
    \end{scope}
		\begin{scope}[xshift=3cm, yscale=0.95, xscale=0.85]
      \node at (0,0) (n0) {$r$};
      \node at (-0.5,-1) (n1) {$a$};
      \node at (0.5,-1) (n2) {$b$};
      \node at (0.5,-2) (n3) {$d$};
      \draw[-,semithick] (n0) -- (n1);
      \draw[-,semithick] (n0) -- (n2);
      \draw[-,double,semithick] (n2) -- (n3);
    \end{scope}
    \draw (m99) edge[<-,bend left,densely dotted] (n0);
    \draw (m0) edge[<-,bend left,densely dotted] (n1);
    \draw (m10) edge[<-,bend right,densely dotted] (n2);
    \draw (m3) edge[<-,bend right,densely dotted] (n3);
    \begin{scope}[xshift = 6cm, yscale=0.95, xscale=0.85]
      \node at (0,0) (x0) {$r$};
      \node at (-1,-1) (x1) {$a$};
      \node at (0,-1) (x2) {$b$};
      \node at (1,-1) (x3) {$c$};
      \node at (-1, -2) (x4) {$a$};
      \node at (1, -2) (x5) {$a$};
      \node at (0, -2) (x6) {$d$};
      \draw[-,semithick] (x1) -- (x0);
      \draw[-,semithick] (x2) -- (x0);
      \draw[-,semithick] (x3) -- (x0);
      \draw[-,semithick] (x4) -- (x2);
      \draw[-,semithick] (x5) -- (x2);
      \draw[-,semithick] (x6) -- (x2);
    \end{scope}
    \draw (x0) edge[<-,bend right,densely dotted] (n0);
    \draw (x1) edge[<-,bend right,densely dotted] (n1);
    \draw (x2) edge[<-,bend left,densely dotted] (n2);
    \draw (x6) edge[<-,bend left,densely dotted] (n3);
  \end{tikzpicture}
\caption{\label{fig:implication}Embeddings of $q$ in $G_S^\exists$ and in a tree $t$ which satisfies $S$ and $q$ at the same time.}
\end{figure}
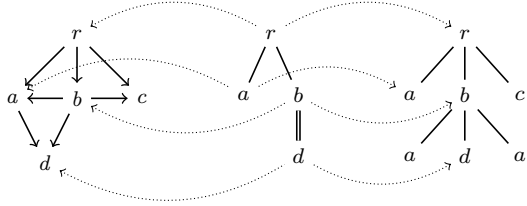
Moreover, we point out that testing the embedding of a query in a dependency graph can be done in polynomial time with a simple bottom-up algorithm. 
From this observation and Lemma~\ref{lemmagraph} we obtain the following.

\begin{theorem}\label{sat-impl-ptime}
$\SAT_{\sims,\Twig}$ and $\IMPL_{\sims,\Twig}$ are in PTIME.
\end{theorem}
Next, we present the complexity of query containment in the presence of IMSs. 
The coNP-completeness of the containment of twig queries~\cite{MiSu04} implies the coNP-hardness of the containment of twig queries in the presence of IMSs. 
Proving the membership of the problem to coNP is, however, not trivial. Given an instance $(p,q,S)$, the set of all trees satisfying $p$ and $S$ can be characterized with a set $\mathcal{G}({p,S})$ containing an exponential number of polynomially-sized graphs and $p$ is contained in $q$ in the presence of $S$ iff the query $q$ can be embedded into all graphs in $\mathcal{G}({p,S})$. This condition is easily checked by a non-deterministic Turing machine.

\begin{theorem}\label{thconp}
$\CNT_{\sims,\Twig}$ is coNP-complete.
\end{theorem}
\begin{proof}
The coNP-completeness of the containment of twig queries (Theorem 4 in~\cite{MiSu04}) implies that $\CNT_{\sims,\Twig}$ is coNP-hard.
Next, we prove the membership of the problem to coNP.
Given an instance $(p,q,S)$, a witness is a function $\lambda:N_p\rightarrow \Sigma$.
Testing whether $\lambda$ is an embedding from $p$ to $G_S^\exists$ requires polynomial time.
If $\lambda$ is an embedding, a non-deterministic polynomial algorithm chooses a graph $G$ from $\Lambda(p,S,\lambda)$ and checks whether $q$ can be embedded in $G$.
We claim that $p\varnot\subseteq_S q$ iff there exists a graph $G$ in $\mathcal G(p,S)$ such that $G\not\preccurlyeq q$.

For the \emph{if} case, we assume that there exists a graph $G\in\mathcal G(p,S)$ such that $G\not\preccurlyeq q$.
We know that $G\preccurlyeq p$, thus $u_G\preccurlyeq p$ (by Lemma~\ref{L3}), hence there exists a tree $t\in L(S)$ such that $t\preccurlyeq p$ and $u_G\unlhd_S t$ (by Lemma~\ref{L6}).
If we assume by absurd that $t\preccurlyeq q$, we have $u_G\preccurlyeq q$, thus $G\preccurlyeq q$, which is a contradiction.
We infer thus that there exists a tree $t\in L(S)\cap  L(p)$, such that $t\notin L(q)$, and consequently, $p\varnot\subseteq_S q$.

For the \emph{only if} case, we assume that $p\varnot\subseteq_S q$, hence there exists a tree $t\in L(S)\cap  L(p)$ such that $t\notin L(q)$.
Because $t\in L(S)\cap  L(p)$, we know that there exists a graph $G\in\mathcal G^* (p,S)$, such that $u_G\unlhd_S t$.
We know that $t\not\preccurlyeq q$, thus $u_G\not\preccurlyeq q$ (by Lemma~\ref{L3}), that yields $G\varnot\preccurlyeq q$.
Furthermore, by using a simple pumping argument, we have
$\forall q\in\Twig.\ \forall G\in\mathcal G^*(q,S).\ (G\not\preccurlyeq q\Rightarrow  \exists G'\in\mathcal G(q,S).\ G'\not\preccurlyeq q)$, which implies that there exists a graph $G'\in\mathcal G(p,S)$ such that $G'\not\preccurlyeq q$.
\qed\end{proof}

\subsection{Extending the complexity results to disjunction-free DTDs}\label{subsec:dtd}
We also point out that the complexity results for implication and containment of twig queries in the presence of IMSs can be adapted to disjunction-free DTDs. 
This allows us to state results which, to the best of our knowledge, are novel.
Similarly to the IMSs, we represent a \emph{disjunction-free DTD} as a tuple $S=(\root_S, R_S)$, where $\root_S$ is a designated root label and $R_S$ maps symbols to regular expressions using no disjunction, basically regular expressions of the grammar:
\[
E ::= \epsilon\mid a\mid E^*\mid E^?\mid E^+\mid (E\cdot E),
\]
\noindent
where $a\in\Sigma$ and ``$\cdot$'' stands for the standard concatenation operator. 
Given such an expression $E$, let $\simbu(E)$ be the set of symbols present in all words from $ L(E)$, and $\simbe(E)$ the set of symbols present in at least one word from $ L(E)$:
\begin{gather*}
\simbu(E) = \{a\in\Sigma\mid \forall w\in L(E).\ \exists w_1,w_2.\ w = w_1\cdot a\cdot w_2\},\\
\simbe(E) = \{a\in\Sigma\mid \exists w\in L(E).\ \exists w_1,w_2.\ w = w_1\cdot a\cdot w_2\}.
\end{gather*}
\noindent\rnew{As pointed out for the IMEs, note that the sets $\simbu(E)$ and $\simbe(E)$ can be easily constructed from $E$.}
Next, we adapt the notions of dependency graph and universal dependency graph for disjunction-free DTDs.
The \emph{existential dependency graph} of a disjunction-free DTD $S$ is a directed rooted graph $G_S^\exists=(\Sigma, \root_S, E_S^\exists)$, where
\[
E_S^\exists = \{(a,a')\mid a'\in\simbe(R_S(a))\}.
\]
Similarly, the \emph{universal dependency graph} of a disjunction-free DTD $S$ is a directed rooted graph $G_S^{\forall}=(\Sigma, \root_S, E_S^{\forall})$, where
\[
E_S^{\forall} = \{(a,a')\mid a'\in\simbu(R_S(a))\}.
\]
Analogously to the IMSs, we assume w.l.o.g.\ that we manipulate only disjunction-free DTDs having no cycle reachable from the root in the universal dependency graph.
Otherwise, if there is a cycle in the universal dependency graph, this means that there is no tree consistent with the schema and containing at least one of the symbols implied in that cycle.
Moreover, similarly to IMSs, for a symbol $a\in\Sigma$ and a disjunction-free regular expression $E$, by $\minnb(E,a)$ we denote the minimum number of occurrences of the symbol $a$ in every word defined by $E$.

Next, we state our complexity results for disjunction-free DTDs.

\begin{theorem}\label{thcor}
$\IMPL_{\mathit{disj\text{-}free}\text{-}\DTD,\Twig}$ is in PTIME and $\CNT_{\mathit{disj\text{-}free}\text{-}\DTD,\Twig}$ is coNP-complete.
\end{theorem}
\begin{proof}
We claim that a query $q$ is implied by a disjunction-free DTD $S$ iff $G_S^{\forall}\preccurlyeq q$ and since the embedding of a query in a graph can be computed in polynomial time, this implies that $\IMPL_{\mathit{disj\text{-}free}\text{-}\DTD,\Twig}$ is in PTIME.
The proof follows from the proof of Lemma~\ref{lemmagraph}.2.
The coNP-completeness of the containment of twig queries (Theorem 4 in~\cite{MiSu04}) implies that $\CNT_{\mathit{disj\text{-}free}\text{-}\DTD,\Twig}$ is coNP-hard.
Theorem~\ref{thconp} states the coNP-completeness of the query containment in the presence of IMSs and an easy adaptation of its proof technique yields the membership of $\CNT_{\mathit{disj\text{-}free}\text{-}\DTD,\Twig}$ to coNP.
The mentioned proofs can be adapted because given a disjunction-free regular expression $E$ and a word $u\in L(E)$, $u$ \rnew{can in fact be }obtained as an ordering of the unordered word $w = \biguplus_{a\in\Sigma} a^{\minnb(E,a)}$.
Moreover, the order imposed by the DTD on the siblings is not important because the twig queries are order-oblivious.
\qed\end{proof}

\section{Expressiveness of DIMS}\label{sec:expressiveness}
First, we compare the expressive power of DIMSs with yardstick languages of unordered trees. 
We begin with FO logic that uses only the binary $\child$ predicate and the unary label predicates $P_a$ with $a\in\Sigma$. 
It is easy to show that DIMSs are not comparable with FO. With a simple rule $a\rightarrow (b\shuffle c)^*$ a DIMS can express the language of trees where every node labeled by $a$ has as children only nodes labeled by $b$ and $c$ such that the number of $b$'s is equal to the number of $c$'s. Such language cannot be captured with FO for reasons similar to those for which it cannot be expressed in FO whether the cardinality of the universe is even. 
There are languages of unordered trees expressible by FO, but not expressible by DIMSs e.g., the language of trees that contain exactly two nodes labeled $b$.
\rnew{Such languages are not expressible by DIMSs for reasons similar to those for which they cannot be expressed by DTDs, more precisely they are not closed under substitution of subtrees with the same root type (cf.\ Lemma 2.10 in~\cite{PaVi00}).}
By using exactly the same examples, note that DIMSs and MSO are also incomparable.
MSO with Presburger constraints~\cite{SeScMu03,SeScMu08,BT05,BTT05} is essentially an extension of MSO that additionally allows elements of arithmetic (numerical variables and value comparisons) and unary functions $\#a$ that return the number of children of a node having a given label $a\in\Sigma$. 
This extension is very powerful and can express Parikh images of arbitrary regular languages. 
DIMSs are strictly less expressive than Presburger MSO as they use a strict restriction of unordered regular expressions.

Next, we compare the expressive power of DIMSs and DTDs. 
For this purpose, we introduce a simple tool for comparing regular expressions with DIMEs. 
Given a regular expression $R$, the language $L(R)$ of unordered words is obtained by removing the relative order of symbols from every ordered word defined by $R$. 
A DIME $E$ \emph{captures} $R$ if $L(E)=L(R)$. 
This tool is equivalent to considering DTDs under \emph{commutative closure}~\cite{BeMi99,NeSc99}. 
We believe that this simple comparison is adequate because if a DTD is to be used in a data-centric application, then supposedly the order between siblings is not important. 
Therefore, a DIME that captures a regular expression defines basically the same admissible content model of a node, without
imposing an order among the children.

Naturally, by using the above notion to compare the expressive powers of DTDs and DIMSs, DTDs are strictly more expressive than DIMSs.
For example, the commutative closure of the regular expression $(a\cdot (b\mid c))^*$ cannot be expressed by a DIME.
Various classes of regular expressions have been reported in widespread use in real-world schemas and have been studied in the
literature: \emph{simple regular expressions}~\cite{BeNeVa04,MaNeSc04}, \emph{single occurrence regular expressions}
(SOREs)~\cite{BNSV10}, \emph{chain regular expressions} (CHAREs)~\cite{BNSV10}. DIMEs are strictly more expressive than CHAREs and incomparable to the other mentioned classes of regular expressions.

Finally, we investigate how many real-life DTDs can be captured with DIMSs and use the comparison on the XMark benchmark~\cite{SWKCMB02} and the University of Amsterdam XML Web Collection~\cite{GrMa11}. 
All 77 regular expressions of the XMark benchmark are captured by DIMEs, and among them 76 by IMEs. 
As for the DTDs from the University of Amsterdam XML Web Collection, $92\%$ of regular expressions are captured by DIMEs and among them $78\%$ by IMEs. 
We also point out that CHAREs, captured by DIMEs, are reported to represent up to 90\% of regular expressions used in real-life DTDs~\cite{BNSV10}. 
These numbers give a generally positive coverage, but should be interpreted with caution, as we do not know which of the considered DTDs were indeed intended for data-centric applications.

\section{Related work}\label{sec:related:work}
Languages of unordered trees can be expressed by \emph{logic formalisms} or by \emph{tree automata}. 
Boneva et al.\ \cite{BT05,BTT05} make a survey on such formalisms and compare their expressiveness. 
The fundamental difference resides in the kind of constraints that can be expressed for the allowed collections of children for some node. 
We mention here only formalisms introduced in the context of XML. 
{\em Presburger automata} \cite{SeScMu03}, {\em sheaves automata} \cite{DaLu03}, and the {\em TQL logic} \cite{CG04} allow to express {\em Presburger constraints} on the numbers of occurrences of the different symbols among the children of some node.
Suitable restrictions allow to obtain the same expressiveness as the \emph{Presburger MSO logic} on unordered trees \cite{BT05,BTT05}, strictly more expressive than DIMSs. 
Additionally, we believe that DIMSs are more appropriate to be used as schema languages, as they were designed as such, in particular regarding the more user-friendly DTD-like syntax.

\rnew{Languages of unordered trees can be also expressed by considering DTDs under \emph{commutative closure}~\cite{BeMi99,NeSc99}.
We assume DTDs using arbitrary regular expressions, not necessarily one-unambiguous~\cite{BuWo98} as required by the W3C.
We also point out that it has been recently shown that it is PSPACE-complete to decide whether a given regular expression can be rewritten as an equivalent one-unambiguous one~\cite{CDLM13}.
Given a DTD using arbitrary regular expressions under commutative closure, we say that an (ordered) tree matches such a DTD iff every tree obtained by reordering of sibling nodes also matches the DTD.
}
However, it is PSPACE-complete to test whether a DTD defines a commutatively-closed set of trees~\cite{NeSc99} and, moreover, such a DTD may be of exponential size w.r.t.\ the size of the alphabet, which makes such DTDs unfeasible.
Another consequence of the high expressive power of DTDs under commutative closure is that the membership problem is
NP-complete~\cite{KoTo10}. 
Therefore, these formalisms were not extensively used in practice. 
From a different point of view, Martens et al.~\cite{MaNe05,MaNeGy08} investigate \rnew{DTDs equipped with} formulas from the $\mathcal{SL}$ logic that specifies unordered languages and obtain complexity improvements for typechecking XML transformations.

The unordered concatenation operator ``$\shuffle$'' should not be confused with the \emph{shuffle} (\emph{interleaving}) operator ``$\&$'' used in a restricted form in XML Schema and RELAX NG to define order-oblivious, \rnew{yet still ordered, content}.
On the one hand, $a^*\&b$ defines all ordered words with an arbitrary number of $a$'s and exactly one occurrence of $b$, and analogously, $a^*\shuffle b$ defines all unordered words with exactly the same characteristic. 
On the other hand, $(a\& b)^*$ defines ordered words of the form $w_1\cdot\ldots\cdot w_n$, where the factors $w_1,\ldots,w_n$ are either $ab$ or $ba$, while $(a\shuffle b)^*$ defines unordered words having the same number of $a$'s and $b$'s. 
For instance, $(a\& b)^*$ does not accept the ordered word $aabb$ while it has the same number of $a$'s and $b$'s.
Adding the shuffle and interval multiplicities to the regular expressions increases the computational complexity of fundamental
decision problems such as: membership~\cite{BeBjHo11,Hovland12}, inclusion, equivalence, and intersection~\cite{GeMaNe09}. 
Colazzo et al.~\cite{CGPS13,CoGhSa09,GhCoSa08} propose efficient algorithms for membership and inclusion of \emph{conflict-free types}, a class of regular expressions with shuffle and numerical constraints using intervals. 
Their approach is based on capturing a language with a set of constraints, similar to our characterizing tuples for DIMEs. 
While conflict-free types and DIMEs both forbid repetitions of symbols, they differ on the restrictions imposed on the use of the operators and the interval multiplicities. 
Consequently, they are incomparable.

\rnew{We finally point out that the static analysis problems involving twig queries i.e., twig query satisfiability~\cite{BeFaGe08}, implication~\cite{HKIF11,BjMaSc13}, and containment~\cite{NeSc06} in the presence of schema have been extensively studied in the context of DTDs.
However, to the best of our knowledge, these problems have not been previously studied neither for the mentioned unordered schema languages, nor for DTDs using classes of regular expressions extended with counting and interleaving.
}

\section{Conclusions and future work}\label{sec:conclusions}
We have studied schema languages for unordered XML.
First, we have investigated languages of unordered words and we have proposed disjunctive interval multiplicity expressions (DIMEs), a subclass of unordered regular expressions for which two fundamental decision problems, membership of an unordered word to the language of a DIME and containment of two DIMEs, are tractable.
Next, we have employed DIMEs to define languages of unordered trees and have proposed disjunctive interval multiplicity schema (DIMS) and its restriction, disjunction-free interval multiplicity schema (IMS).
DIMSs and IMSs can be seen as DTDs using restricted classes of regular expressions and interpreted under commutative closure to define unordered content models.
These restrictions allow to maintain a relatively low computational complexity of basic static analysis problems while allowing to capture a significant part of the expressive power of practical DTDs.

\rnew{As future work, we want to study whether the restrictions imposed by the grammar of DIMEs can be relaxed while maintaining the tractability of the problems of interest.}
Moreover, we would like to investigate learning algorithms for the unordered schema languages proposed in this paper.
We have already proposed learning algorithms for restrictions of DIMSs and IMSs~\cite{CiSt13} and we want to extend them to take into account all the expressive power.
We also aim to apply the unordered schemas to query minimization~\cite{ACLS02} i.e., given a query and a schema, find a smaller yet equivalent query in the presence of the schema. 
Furthermore, we want to use unordered schemas and optimization techniques to boost the learning algorithms for twig queries~\cite{StWi12}.

\bibliographystyle{plain}
\bibliography{dime}
\end{document}

%% file: header.tex
\usepackage{amsmath}
\usepackage{amssymb}
\usepackage{mathabx}
\usepackage{colonequals}
\usepackage{tikz}
\usepackage{subfigure}
\usepackage{verbatim}
\usepackage{float}
\usepackage{xspace}
\usepackage{url}
\usepackage{ifthen}
\usepackage{stmaryrd}
\usepackage{enumerate}
\usepackage{enumitem}

\newcommand{\rrc}{\color{black}}
\newcommand{\rmargin}[1]{%

  \marginpar{%
    \rrc%
    \raggedright%
    \fontsize{7}{7.5}%
    \selectfont%

\ifthenelse{\isodd{\thepage}}
   {
\begin{tabular}{cc}
~~~~& \begin{minipage}{2.2cm}#1\end{minipage}
\end{tabular}
	}
  {#1}
  }
}

\newcommand{\rnew}[1]{{\rrc #1}}

\def\qed {{\parfillskip=0pt \widowpenalty=10000 \displaywidowpenalty=10000 \finalhyphendemerits=0 \leavevmode \unskip \nobreak \hfil\penalty50 \hskip.2em \null \hfill $\square$ \par}}

\newcommand{\Nz}{\ensuremath{\mathbb{N}_0}}

\newcommand{\Paths}{\ensuremath{\mathit{Paths}}}
\newcommand{\LabPaths}{\ensuremath{\mathit{LabPaths}}}
\newcommand{\shuffle}{\mathbin{|\hspace{-0.1em}|}}
\newcommand{\dblslash}{\ensuremath{/\!/}}
\newcommand{\wc}{\mathord{\star}}

\newcommand{\lab}{\mathit{lab}}
\renewcommand{\root}{\mathit{root}}

\newcommand{\child}{\mathit{child}}
\newcommand{\ch}{\mathit{ch}}

\newcommand{\desc}{\mathit{desc}}
\newcommand{\simbu}{\rnew{\mathit{symbols}^\forall}}
\newcommand{\simbe}{\rnew{\mathit{symbols}^\exists}}
\newcommand{\minnb}{\mathit{min\_nb}}
\newcommand{\true}{\mathit{true}}
\newcommand{\false}{\mathit{false}}
\newcommand{\open}{\mathit{open}}
\newcommand{\close}{\mathit{close}}
\newcommand{\Twig}{\ensuremath{\mathit{Twig}}}
\newcommand{\Tree}{\ensuremath{\mathit{Tree}}}
\newcommand{\CNT}{\mathrm{CNT}}
\newcommand{\SAT}{\mathrm{SAT}}
\newcommand{\IMPL}{\mathrm{IMPL}}

\newcommand{\MEMB}{\mathrm{MEMB}}
\newcommand{\DTD}{\mathit{DTD}}
\newcommand{\sims}{\mathit{IMS}}
\newcommand{\dims}{\mathit{DIMS}}

\newcommand{\counta}{\mathit{count}}
\newcommand{\readb}{\mathit{read}}
\newcommand{\impl}{\mathit{impl}}
\newcommand{\pmin}{\mathit{P^{\subseteq_{\tiny\min}}_E}}
\newcommand{\pminp}{\mathit{P^{\subseteq_{\tiny\min}}_{E'}}}
\newcommand{\pminpp}{\mathit{P^{\subseteq_{\tiny\min}}_{E_1}}}

\newcommand{\satt}{\mathrm{SAT}_{1\mbox{-}\textrm{in}\mbox{-}3}}
\newcommand{\pip}{\Pi_2^{\mathrm P}}
\newcommand{\qbf}{\forall^*\exists^*\mathrm{QBF}}

%% file: alg.tex
\floatstyle{ruled}
\newfloat{float@algorithm}{htbp}{loa}
\floatname{float@algorithm}{Algorithm}
\newcounter{LineCounter@algorithm} 

\newenvironment{BasicCommands@algorithm}{%
  \newcommand{\TAB}{\makebox[4ex][r]{}}%
  \newcommand{\ALGORITHM}{\textbf{algorithm}\xspace}%
  \newcommand{\INPUT}{\textbf{Input}\xspace}%
  \newcommand{\OUTPUT}{\textbf{Output}\xspace}%
  \newcommand{\LET}{\textbf{let}\xspace}%
  \newcommand{\IF}{\textbf{if}\xspace}%
  \newcommand{\THEN}{\textbf{then}\xspace}%
  \newcommand{\WHILE}{\textbf{while}\xspace}%
  \newcommand{\FOR}{\textbf{for}\xspace}%
  \newcommand{\DO}{\textbf{do}\xspace}%
  \newcommand{\RETURN}{\textbf{return}\xspace}%
}{%
}

\newenvironment{algorithm*}%
{%
\begin{BasicCommands@algorithm}%
\list{}{\itemindent 0em%
        \listparindent\itemindent
        \rightmargin  \leftmargin}%
\item\relax
}{%
\endlist
\end{BasicCommands@algorithm}%
}

\newenvironment{algorithm}%
{%
  \newcommand{\ResetLineCounter}{%
    \setcounter{LineCounter@algorithm}{0}%
  }%
  \newcommand{\LN}{%
    \makebox[0pt][l]{%
      \makebox[0pt][l]{%
        \addtocounter{LineCounter@algorithm}{1}%
      }%
      \makebox[10.75pt][r]{
        \fontsize{7}{10}%
        \selectfont%
        \arabic{LineCounter@algorithm}:%
      }%
    }%
    \TAB%
  }%
  \begin{BasicCommands@algorithm}%
  \begin{float@algorithm}%
    \ResetLineCounter%
}%
{%
  \end{float@algorithm}%
  \end{BasicCommands@algorithm}%
}